\documentclass[11pt,letterpaper]{article}

\usepackage{color}
\usepackage{url}
\usepackage{makeidx}
\usepackage{amssymb}
\usepackage{epsfig}
\usepackage{epstopdf}
\usepackage{subfig}
\usepackage[square,sort,comma,numbers]{natbib}
\usepackage[ruled,vlined,linesnumbered,commentsnumbered]{algorithm2e}
\usepackage{amsthm}

\newtheorem{definition}{Definition}

\newtheorem{lemma}{Lemma}

\DontPrintSemicolon
\SetAlgoLined

\newcommand{\remove}[1]{}

\begin{document}

\title{Enumerating Maximal Bicliques from a Large Graph using MapReduce
\footnote{A preliminary version of the paper ``Enumerating Maximal Bicliques from a Large Graph using MapReduce" was accepted at the Proceedings of the 3rd IEEE International Congress on Big Data 2014. }
}

\author{
Arko Provo Mukherjee\thanks{
Department of Electrical and Computer Engineering, Iowa State University.
Email: arko@iastate.edu
}
\and 
Srikanta Tirthapura\thanks{
Department of Electrical and Computer Engineering, Iowa State University.
Email: snt@iastate.edu
}
\thanks{The research reported in this papers is partially supported by
  the HPC equipment purchased through NSF MRI grant number CNS 1229081
  and NSF CRI grant number 1205413.}  
}

\maketitle 

\begin{abstract}
We consider the enumeration of maximal bipartite cliques (bicliques)
from a large graph, a task central to many practical data mining
problems in social network analysis and bioinformatics. We present
novel parallel algorithms for the MapReduce platform, and an
experimental evaluation using Hadoop MapReduce.

Our algorithm is based on clustering the input graph into smaller
sized subgraphs, followed by processing different subgraphs in
parallel. Our algorithm uses two ideas that enable it to scale to
large graphs: (1)~the redundancy in work between different subgraph
explorations is minimized through a careful pruning of the search
space, and (2)~the load on different reducers is balanced through the
use of an appropriate total order among the vertices. Our evaluation
shows that the algorithm scales to large graphs with millions of edges
and tens of millions of maximal bicliques. To our knowledge, this is
the first work on maximal biclique enumeration for graphs of this
scale.
\end{abstract}


\section{Introduction}
\label{sec:intro}

A graph is a natural abstraction to model rich relationships in data,
and massive graphs are ubiquitous in applications such as online
social networks~\cite{MMGDB2007,NWS2002}, information retrieval from
the web~\cite{BKMRRSTW2000}, citation networks~\cite{AJM2004}, and
physical simulation and modeling~\cite{WTCG2012}, to name a few.
Finding information from such data can often be reduced to a problem
of mining features from massive graphs.
We consider scalable methods for discovering densely connected
subgraphs within a large graph. Mining dense substructures such as
cliques, quasi-cliques, bicliques, quasi-bicliques etc. is an
important area of study~\cite{AACFHS2004,GKT2005,ARS2002,SLGL2006}.

A fundamental dense substructure of interest is a {\em biclique}. A
biclique in a graph $G=(V,E)$ is a pair of subsets of vertices
$L\subseteq V$ and $R\subseteq V$ such that (1)~$L$ and $R$ are
disjoint and (2)~there is an edge $(u,v) \in E$ for every $u \in L$
and $v \in R$. For instance, consider the following graph relevant to
an online social network, where there are two types of vertices, users
and webpages. There is an edge between a user and every webpage that
the user ``likes'' on the social network. A biclique in such a graph
consists of a set of users $U$ and a set of webpages $W$ such that
every user in $U$ has liked every page in $W$. Such a biclique
indicates a set of users who share a common interest, and is valuable
for understanding the actions of users on this social network. Often,
it is useful to identify only {\em maximal bicliques} in a graph,
which are those bicliques that are not contained within any other
larger bicliques. {\em We consider the problem of enumerating all
maximal bicliques from a graph (henceforth referred to as MBE)}.

Many graph mining tasks have relied on enumerating bicliques to
identify significant substructures within the graph. For instance the
analysis of web search queries~\cite{YM2009} considered the
``click-through'' graph, where there are two types of vertices, web
search queries and web pages. There is an edge from a search query to
every page that a user has clicked in response to the search
query. MBE was used in clustering queries using the click through
graph. MBE has been used in social network analysis, in detection of
communities in social networks~\cite{LSH2008}, and in finding
antagonistic communities in trust-distrust
networks~\cite{LSZL2011}. It has also been applied in detecting
communities in the web graph~\cite{KRRT1999,RH2005}.

In bioinformatics, MBE has been used widely e.g. in construction of
the phylogenetic tree of
life~\cite{DABMMS2004,SDREL2003,YBE2005,NK2008}, structure discovery
and analysis of protein-protein interaction
networks~\cite{BZCXZLZSLZLC2003,SLL2011}, analysis of gene-phenotype
relationships~\cite{XPH2012}, prediction of miRNA regulatory
modules~\cite{YM2005}, modeling of hot spots at protein
interfaces~\cite{LL2009}, and in analysis of relationships between
genotypes, lifestyles, and diseases~\cite{MKGR2007}. Other
applications include Learning Context Free Grammars~\cite{Y2011},
finding correlations in databases~\cite{J2005}, for data
compression~\cite{AAAAAASS1993}, role mining in role based access
control~\cite{CPOV2010}, and process operation
scheduling~\cite{MGP2011}.

\begin{figure}
\begin{center}
\includegraphics[width=0.5\textwidth]{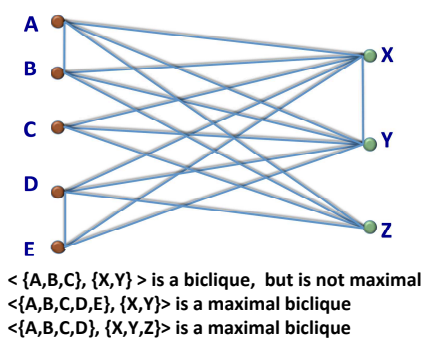}
\caption{
\label{fig:max_biclique}
Maximal Bicliques}
\end{center}
\end{figure}

While it is easy to find a single maximal biclique in a graph,
enumerating all maximal bicliques is an NP-hard problem
(Peters~\cite{P2003}). This does not however mean that typical cases
are unsolvable. In fact, there are output-polynomial time algorithms
whose theoretical runtime is bounded by a polynomial in the number of
vertices in the graph, and the number of maximal bicliques that are
output~\cite{AACFHS2004}. Thus it is reasonable to expect to be able
to devise algorithms for MBE that work on large graphs, as long as the
number of maximal bicliques output is not too high. 

\remove{
Typically most large graphs do not exhibit exponentially large output sizes.  
}

Current methods for enumerating bicliques have the following
drawbacks. Most methods are sequential algorithms that are unable to
make use of the power of multiple processors. For handling large
graphs, it is imperative to have methods that can process a graph in
parallel. Next, they have been evaluated only on small graphs of a few
thousands of vertices and a few hundred thousand maximal bicliques,
and have not been shown to scale to large graphs. For instance, the
popular ``consensus'' method for biclique
enumeration~\cite{AACFHS2004} presents experimental data only on
graphs of up to 2,000 vertices, and about 140,000 maximal bicliques,
and other works~\cite{LLSW2005,LSL2006} are also similar.
\footnote{In our experiments, we show that the consensus and other 
sequential methods are unable to process our input graphs in a 
reasonable time}.
\footnote{It is not possible to quantify the
complexity of a problem instance through the input size (number of
vertices,and edges). However, the number of maximal bicliques, used in
conjunction with the input size, is more indicative of the
complexity.}

{\em Our goal is to design a parallel
method that can enumerate maximal bicliques in large graphs, with
millions of edges and tens of millions of maximal bicliques, and which
can scale with the number of processors.}

\subsection{Contributions}
\label{sec:contribution}

\remove{
We present a framework for a parallel solution to MBE using a
clustering approach that divides the input graph into overlapping
subgraphs that can be processed independently. We implement this
framework on MapReduce~\cite{DG2008}, and evaluate different algorithms
within the framework, based on different sequential algorithms for
MBE.
}

We present a parallel solution for MBE using the MapReduce
framework~\cite{DG2008}. At a high level, our approach clusters the
input graph into overlapping subgraphs that can be processed
independently in parallel, by different reducers.
We implement the clustering approach using two different
state-of-the-art sequential algorithms for MBE, one based on depth
first search~\cite{LSL2006}, and the other based on the consensus
algorithm~\cite{AACFHS2004}.  

For this clustering approach to be effective on large graphs, we
needed to augment it with two ideas that significantly improve the
parallel performance. The first idea is concerned with reducing the
overlap in the work done by different subtasks. It is usually not
possible to assign disjoint subgraphs to different processors, and the
subgraphs assigned to different tasks will overlap, sometimes
significantly. Through a careful partitioning of the search space
among the different tasks, we reduce redundant work among the tasks
(this partitioning depends on details of the sequential algorithm used
at each task).

The second idea is concerned with balancing the load between different
tasks. With a graph analysis task such as biclique enumeration, the
complexity of different subgraphs varies significantly, depending on
the density of edges in the subgraph. Naively done, this can lead to a
case where most reducers finish quickly, while only a few take a long
time, leading to a poor parallel performance. We present a solution to
keep the load more balanced, based on an ordering of vertices, which
reduces enumeration load on subgraphs that are dense, and increases
the load on subgraphs that are sparse, leading to a better parallel
efficiency. 
We provide some basic statistical analysis of the Reducer
runtimes with and without the load balancing to justify our claim.

We give a detailed analysis of the communication costs of the 
clustering based MapReduce Algorithms described.


We present a direct parallelization of the consensus sequential
algorithm~\cite{AACFHS2004}, using an approach different from
clustering. We found that while this approach may use a smaller memory
per node that the clustering approach, it requires substantially
greater runtime.

Finally, we present detailed experimental results on real-world and
synthetic graphs. Overall, the clustering approach (using depth first
search), when combined with load balancing and reduction of redundant
work, performs the best on large graphs. Our algorithms can process
graphs having millions of edges, and tens of millions of maximal
bicliques, and can scale out with the cluster size. To our knowledge,
these are the largest reported graph instances where bicliques have
been successfully enumerated.

We also provide experimental evidence showing that our parallel
Algorithm based on depth first search effectively generates only large
maximal bicliques. For this we show that the runtime of our parallel
algorithm decreases with increase in the size threshold of the generated
maximal bicliques.


\subsection{Prior and Related Work}
\label{sec:literature}
Makino {\em et. al}~\cite{MU2004} describes methods to enumerate all
maximal bicliques in a {\em bipartite graph}, with the delay between
outputting two bicliques bounded by a polynomial in the maximum degree
of the graph. Zhang {\em et. al}~\cite{YEM2008} describe a
branch-and-bound algorithm for the same problem. However, these
approaches do not work for general graphs, as we consider here.

\remove{
describe an
output-polynomial time branch-and-bound algorithm for the same problem
search algorithm having
runtime $O(d \cdot n^2 \cdot \beta )$, where $d$ is the maximum degree
of any vertex, $n$ is the number of vertices and $\beta$ is the number
of maximal bicliques.  Note that these approaches do not work for
general graphs, like we consider here.
}

There is a variant of MBE where we only seek {\em induced} maximal
bicliques in a graph. An induced maximal biclique is a maximal
biclique which is also an induced subgraph; i.e. a maximal biclique
$\langle L,R \rangle$ in graph $G$ is an induced maximal biclique if
$L$ and $R$ are themselves independent sets in $G$.  We consider the
{\em non-induced} version, where edges are allowed in the graph
between two vertices that are both in $L$, or both in $R$ (such edges
are of course, not a part of the biclique). The set of maximal
bicliques that we output will also contain the set of induced maximal
bicliques, which can be obtained by post-processing the output of our
algorithm.  Note that for a bipartite graph, every maximal biclique is
also an induced maximal biclique. Algorithms for Induced MBE include 
work by Eppstein~\cite{E1994}, Dias {\em et. al}~\cite{DFS2005}, and 
Gaspers {\em et. al}~\cite{GKL2008}.

\remove{
 Eppstein~\cite{E1994} gave an
algorithm for induced MBE, which runs in polynomial time for certain
special classes of graphs with bounded arboricity. Dias {\em
et. al}~\cite{DFS2005} present a method for induced MBE, where the
bicliques are output in lexicographic order with polynomial time
delay.  Gaspers {\em et. al}~\cite{GKL2008} present an Algorithm for
induced MBE with runtime $O(1.3642^n)$, where $n$ is the number of
vertices in the graph.
}


\remove{
}

Alexe {\em et. al}~\cite{AACFHS2004} present an iterative algorithm
for non-induced MBE using the ``consensus'' method, which 
we briefly review in Section~\ref{sec:algo}. 
%
Another technique for MBE is based on a recursive depth first search
(DFS)~\cite{LLSW2005,LSL2006}. \cite{LLSW2005} presents an approach
based on a connection with mining closed patterns in a transactional
database, and apply the algorithm from~\cite{UKA2004}, which is based
on depth first search. \cite{LSL2006} present a more direct algorithm
for biclique enumeration based on depth first search, which we use in
our work. This is described in more detail in Section~\ref{sec:algo}.

\remove{
The first work using the DFS approach is based on mining closed
patterns~\cite{LLSW2005}.  The authors provide a modified version of
the DFS LCM algorithm~\cite{UKA2004} which they call LCM-MBC.  The
authors claim that the time complexity of the proposed algorithm is
$O(n \cdot m \cdot N)$ and the space complexity is $O(n \cdot m)$
where $n$ is the number of vertices, $m$ is the number of edges and
$N$, the number of maximal bicliques.  We also have the direct
depth-first search algorithm by Liu {\em et. al}~\cite{LSL2006}.  This
algorithm has a runtime of $O(n \cdot d \cdot N)$ where $n$ is the
number of vertices, $d$ is the maximum degree of any vertex and $N$,
the number of maximal bicliques.  The space complexity of the
algorithm is $O(m + d^2)$.  This algorithm is also described in
Section~\ref{sec:algo}.  We used the second DFS approach rather than
the first one, as the direct DFS approach has a better runtime bound
and also as in the first approach we need to transform the adjacency
matrix of the graph to a transactional database.  For a large graph,
we wanted to avoid using the adjacency matrix as it requires more
space than the adjacency list.
}

Another approach to MBE is through a reduction to the problem of
enumerating maximal cliques, as described by G\'ely {\em
et. al}~\cite{GNS2009}. Given a graph $G$ on which we need to
enumerate maximal bicliques, a new graph $G'$ is derived such that
through enumerating maximal cliques in $G'$ using an algorithm such
as~\cite{TTT2006,TIAS1977}, it is possible to derive the maximal
bicliques in $G$. However, this approach is not practical for large
graphs since in going from $G$ to $G'$, the number of edges in the
graph increases significantly.


\remove{
we define a new graph $G_1 =
(V_1,E_1)$ as follows: Vertex set
$V_1= V\cup V'$, where $V'$ is a copy of vertices in $V$, and
edge set $E_1 = \{ uv' \in V \times V' \mid v \in V, uv \in E \} \cup
\{ uv \in V^2, uv \not \in E \} \cup \{ u'v' \in {V'}^2, uv \not \in E\}$.
It can be shown that $B = \langle L,R \rangle $ is a maximal
biclique in $G$ only if $\{ L \cup R' \}$ and $\{ L' \cup R \}$ are
cliques in $G_1$. Thus, 
Maximal clique enumeration algorithms~\cite{TTT2006,TIAS1977} can be
used for MBE.
}

To our knowledge, the only prior work on parallel algorithms for MBE
is by Nataraj and Selvan~\cite{NS2009}, who use the correspondence
between maximal bicliques and closed patterns~\cite{LLSW2005} to
derive a parallel method for enumerating maximal bicliques. A
significant issue is that~\cite{NS2009} assumes that the input graph
is presented as an adjacency matrix, which is then converted into a
transactional database and distributed among the processors. In
contrast, we do not assume an adjacency matrix, but assume that the
graph is presented as a list of edges. Thus we are able to work on
much larger graphs than \cite{NS2009}; the largest graph that they
consider has 500 vertices and about 9000 edges.

\remove{
For our work, we take a more direct 
approach towards solving the problem and apply depth first search
technique directly on the input graph without converting it into a
transactional database. 
The input to the Algorithm given in~\cite{NS2009} is the
adjacency matrix of the graph. By not converting the graph
into a transactional database, we avoid using the 
adjacency matrix of the graph. This is certainly helpful
for large graphs, as the adjacency matrix can be quite large.
Further, unlike~\cite{NS2009}, our work uses the popular 
MapReduce framework~\cite{DG2008} which has advantages like 
fault tolerance, data distribution etc. and can be run on a
commodity cluster. We show our results on much larger graphs,
some artifically created, while some taken from the real world.
}

MBE is related to, but different from the problem of finding the
largest sized biclique within a graph (maximum biclique). There are a
few variants of the maximum biclique problem, including maximum edge
biclique, which seeks the biclique in the graph with the largest
number of edges, and maximum vertex biclique, which seeks a biclique
with the largest number of edges; for further details and variants,
see Dawande {\em et al.}~\cite{DKST2001}. MBE is harder than finding a
maximum biclique, since it enumerates all maximal bicliques, including
all maximum bicliques.



\section{Preliminaries}
\label{sec:prelim}
We present a formal problem definition, review prior sequential
algorithms, and then briefly review the MapReduce
parallel programming model that we use.


\subsection{Problem Definition}
\label{sec:probdef}
We consider a simple undirected graph $G=(V,E)$ without self-loops or
multiple edges, where $V$ is the set of all vertices and $E$ is the
set of all edges of the graph. Let $n = \left|{V}\right|$ and 
$m = \left|{E}\right|$.
Graph $H=(V_{1},E_{1})$ is said to be a sub-graph of graph $G=(V,E)$ if
$V_{1} \subset V$ and $E_{1} \subset E$. $H$ is known as an
induced subgraph if $E_{1}$ consists of all edges of $G$ that connect
two vertices in $V_{1}$. For vertex $u \in V$, let $\eta(u)$ denote
the vertices adjacent to $u$. For a set of vertices 
$U \subseteq V$, let $\eta(U) = \bigcup \limits _{u \in U} \eta(u)$.
For vertex $u \in V$ and $k > 0$, let $\eta^{k}(u)$ denote all
vertices that can be reached from $u$ in $k$ hops. For 
$U \subseteq V$, let 
$\eta^{k}(U) = \bigcup \limits _{u \in U} \eta^{k}(u)$. 
We call $\eta^{k}(U)$ as the $k$-neighborhood of $U$.
For a set of vertices 
$U \subseteq V$, let $\Gamma(U) = \bigcap \limits _{u \in U} \eta(u)$.


\begin{definition}
\label{def:biclique}
A \emph{biclique} $B = \langle L,R \rangle$ is a subgraph of $G$
containing two non-empty and disjoint vertex sets, $L$ and $R$ such that for any
two vertices $u \in L$ and $v \in R$, there is an edge $(u,v) \in E$.
\end{definition}

Note that the definition on $B=\langle L,R \rangle$ does not impose
any restriction on the existence of edges among the vertices within
$R$ or within $L$, i.e., we consider \emph{non-induced} bicliques.

\begin{definition}
\label{def:maxbiclique}
A biclique $M=\langle L,R \rangle$ in $G$ is said to be a 
maximal biclique if there is no other biclique 
$M'=\langle L',R'\rangle \neq \langle L,R \rangle$
such that $L \subset L'$ and $R \subset R'$.
\end{definition}

\remove{
\begin{definition}
\label{def:star}
A biclique $B = \langle L,R \rangle$ is called a {\em star} if the
left (or right) set is a single vertex $v \in V$ and the other set is
$\eta(v)$.
\end{definition}
}

The {\em \bf Maximal Biclique Enumeration Problem (MBE)} is to
enumerate the set of all maximal bicliques in graph $G=(V,E)$.

In our algorithms, we assume that vertex identifiers are unique and
are chosen from a totally ordered set. This is usually not a limiting
assumption, since vertex identifiers are usually strings which can be
ordered using the lexicographic ordering.

\subsection{Sequential Algorithms}
\label{sec:algo}
We describe the two approaches to sequential algorithms for MBE that
we consider, one based on a ``consensus algorithm''~\cite{AACFHS2004},
and the other based on depth first search~\cite{LSL2006}.

\paragraph{Consensus Algorithm}
Alexe {\em et. al}~\cite{AACFHS2004} present an iterative approach to
MBE. This type of algorithm starts off with a set of simple ``seed''
bicliques.  In each iteration, it performs a ``consensus'' operation,
which involves performing a cross-product on the set of current
candidates bicliques with the seed bicliques, to generate a new set of
candidates, and the process continues until the set of candidates does
not change anymore. Due to lack of space, we do not present the
details here, and refer the reader to~\cite{AACFHS2004}. It is proved
that these algorithms exactly enumerate the set of maximal bicliques
in the input graph.

The consensus approach has a good theoretical guarantee, since its
runtime depends on the number of maximal cliques that are output.  In
particular, the runtime of the MICA version of the algorithm is proved
to be bounded by $O\left (n^3 \cdot N \right)$ where $n$ is the number
of vertices and $N$ total number of maximal bicliques in $G$.
The consensus algorithm has been found to be adequate for many
applications and is quite popular.

We use the consensus algorithm in two ways. One as a candidate method
for a sequential algorithm within each cluster. In another, we
consider a direct parallelization of the consensus algorithm without
using the clustering method.

\paragraph{Sequential DFS Algorithm}

The basic sequential depth first approach (DFS) that we use is
described in Algorithm~\ref{SPA}, based on~\cite{LSL2006}. It attempts
to expand an existing maximal biclique into a larger one by including
additional vertices that qualify, and declares a biclique as maximal
if it cannot be expanded any further. The algorithm takes the
following inputs: (1)~the graph $G=(V,E)$, (2)~the current
vertex set being processed, $X$, (3)~$T$, the tail vertices of
$X$, i.e. all vertices that come after $X$ in lexicographical
ordering and
(4)~$s$, the minimum size threshold below which a 
maximal biclique is not enumerated. $s$ can be set to $1$ so as to 
enumerate all maximal bicliques in the input graph.
However, we can set $s$ to a larger value to enumerate only large
maximal bicliques such that for $B = <L,R>$, we have 
$ \left | L \right | \ge s$ and $ \left | R \right | \ge s$.
The size threshold $s$ is provided as user input.
The other inputs are initialized as follows: $X=\varnothing$, $T=V$.

The algorithm recursively searches for maximal bicliques. It  increases
the size of $X$ by recursively 
adding vertices from the tail set $T$, and pruning away 
those vertices from $T$ which along with $X$ 
do not have any any common vertices in their neighborhood.
From the expanded $X$, the algorithm outputs the maximal biclique
$\left < \Gamma(\Gamma(X)), \Gamma(X) \right >$.

\begin{algorithm} 
\caption{Depth First Search: PA($G$,$X$,$T$,$s$)} 
\label{SPA}

	\ForAll{vertex $v \in T$}
	{
		\If{$\left | \Gamma(X \cup \{ v \}) \right | < s$}
		{
			$T \leftarrow T \setminus \{ v \}$ \;
		}
	}
	
	\If{$\left | X \right | + \left | T \right | < s$}
	{
		\Return
	}
	
	Sort vertices in $T$ as per ascending order of $\left | \eta(X \cup \{ v \}) \right |$ \;
	
	\ForAll{vertex $v \in T$}
	{
		$T \leftarrow T \setminus \{ v \}$ \;
		
		\If{$\left | X \cup \{ v \} \right | + \left | T \right | \ge s$}
		{

			$N \leftarrow \Gamma( X \cup \{ v \} )$ \;
			$Y \leftarrow \Gamma( N )$ \; 
			Biclique $B \leftarrow \left < Y, N \right >$ \;

			\If{$ ( Y \setminus ( X \cup \{ v \} ) ) \subseteq T $}
			{
				\If{$\left | Y \right | \ge s$}
				{
					Emit $B$ as a maximal biclique \; 
				}	

				PA($G$, $Y$, $T \setminus Y$, $s$) \;
			}

		}
	}
\end{algorithm} 

\subsection{MapReduce}
\label{sec:model}

MapReduce~\cite{DG2008} is a popular framework for processing large
data sets on a cluster of commodity hardware. A MapReduce program is
written through specifying map and reduce functions. The map
function takes as input a key-value pair $\langle k,v\rangle $ and
emits zero, one, or more new key-value pairs $\langle k',v'\rangle$. 
All tuples with the same value of the key are grouped together and
passed to a reduce function, which processes a particular key $k$
and all values that are associated with $k$, and outputs a final list
of key-value pairs. The outputs of one MapReduce round can be the
input to the next round. 
Communication happens only when the outputs from the map methods
are retrieved by the different reduce methods based on the key, i.e.
when data is grouped by keys.  Further details are available
in~\cite{DG2008,GGL2003}.  We used Hadoop~\cite{W2009,SKRC2010}, an
open source implementation of MapReduce, on top of a distributed file
system HDFS.  While we consider the MapReduce framework for this work,
our algorithms are generic and can be used with other distributed
frameworks like Pregel~\cite{MABDHLC2010}.

\remove{
The input to the
algorithm is processed by the map method and it emits a sequence of
key / value pairs $\langle k_i,v_i \rangle$.  Then the supporting
system groups the values according to the keys and sorts them. Thus
for each key $k_{i} \subset K$, the system generates $ \langle
k_{j},list(v_{j}) \rangle $. This is then passed on to the reduce
method as the input. The output from the reduce method can then be
used as the input to the next round.  Multiple map instances (all part
of the same round) can run in parallel, as can multiple reduce
instances.}


\section{MapReduce Algorithms for MBE}

In this section, we describe algorithms for MBE using MapReduce.  We
first present the basic clustering approach, which can be used with
any sequential algorithm for MBE, followed by enhancements to the
basic clustering approach, and finally the parallel consensus approach.

\subsection{Basic Clustering Approach}
\label{sec:part}
We first present the basic clustering framework for parallel MBE. The
approach is to cluster the input graph into several overlapping
sub-graphs (clusters) and then run the sequential DFS algorithm in
parallel for each cluster.

For each $v\in V$, the cluster $C(v)$ consists of the induced subgraph
on all vertices in $\eta^2(v)$ (i.e. the 2-neighborhood of $v$ in
$G$). The different clusters $C(v)$ are processed in parallel, and a sequential MBE
algorithm is used to enumerate the maximal bicliques from each
cluster. While all maximal bicliques in $G$ are indeed output by this
approach, the same biclique maybe enumerated multiple times.  To
suppress duplicates, the following strategy is used: a maximal
biclique $B$ arising from cluster $C(v)$ is emitted only if $v$ is the
smallest vertex in $B$ according to the total order of the vertices.
The basic clustering framework is generic and can be used with any
sequential algorithm for MBE. We have considered the clustering
algorithm using the DFS and the consensus algorithms for MBE.


\remove{
For a given input graph $ G = (V,E) $ where 
$ n = \left|{V}\right| $ we construct $n$ clusters. For each vertex $v$ in the
graph, we create a cluster consisting of the induced sub--graph created by
the 2--neighborhood of $v$.
For a vertex $v$, we call the corresponding cluster as $C(v)$.
Each such cluster $C(v)$ is processed by a separate reducer.
It can be shown that all maximal bicliques can be enumerated using this
approach. However, since each cluster is processed independently,
any maximal biclique will get generated for multiple clusters.
For a maximal biclique $M= \langle L,R \rangle$, let $V_M$ represent all the vertices in $M$.
It can  be shown that the maximal biclique $M$ will get enumerated
by all the clusters $C(v)$ such that $v \in V_M$.
We eliminate the duplicate maximal bicliques by the following strategy :
A maximal biclique $M$ is emitted by a cluster $C(v)$, only if
$v$ is the smallest vertex in $M$.
}

\begin{lemma} 
\label{lemma:part} 
The basic clustering approach enumerates all maximal bicliques in 
graph $G=(V,E)$.
\end{lemma}

\begin{proof}
  We show the following two properties. First, every maximal biclique
  in $G$ must be output as a maximal biclique from cluster $C(v)$ for
  some $v \in V$. Second, every maximal biclique output from each
  cluster must be a maximal biclique in $G$.
  To prove the first direction, consider a maximal biclique $M=\langle
  L,R\rangle$ in $G$. Let $v$ be the smallest vertex in $M$ in
  lexicographic order, and without loss of generality suppose that $v
  \in L$. By the definition of a biclique, for each $u \in R$, $u$ is
  a neighbor of $v$. Similarly, every vertex $w \in L$ is a neighbor
  of $v$, and is hence in $\eta^2(v)$. Hence $M$ is
  completely contained in $C(v)$. Note that $M$ is also a maximal
  biclique in $C(v)$. To see this, note that if $M$ is not maximal
  biclique in $C(v)$, then $M$ is not maximal in $G$ either.

  We prove by contradiction that every maximal biclique in each
  cluster $C(v)$ is also a maximal biclique in $G$. Consider a
  biclique $M$ emitted as maximal from cluster $C(v)$ such that it is
  not maximal in $G$. Then, there exists a maximal biclique $M'$ that
  can be generated by extending $M$. However, it is easy to see that
  every vertex in $M'$ must also be contained in $\eta^2(v)$, 
and hence $M'$ is also contained in $C(v)$, contradicting our
  assumption that $M$ is a maximal biclique in $C(v)$.
\end{proof}

\remove{
Without loss in generality, let $v \in M$ be
  the smallest vertex in $M$.  and let $v$ be in the left set $L$. Now
  from the definition of a biclique for all vertices $u \in R$, $u$
  must be a neighbor of $v$ i.e. $u$ must be in 1-neighborhood of $v$.
  Similarly, any vertex $w \in L$ must be in 1-neighborhood of $u$.
  This implies that for any vertex $w \in L$, $w$ must be in
  2-neighborhood of $v$. Hence $M$ will be a biclique in $C(v)$.
  Since $M$ is a maximal biclique in $G$, and $C(v)$ is an induced
  sub--graph of $G$, $M$ will also be a maximal biclique in $C(v)$.
  This is because, if there exists a vertex $z \in C(v)$ that can
  extend $M$, then, $M$ will not be a maximal biclique in $G$.

  Next we show the second direction. 
  That is we need to show that every maximal biclique in a cluster
  is also a maximal biclique in $G$.
  We prove by contradiction.
  Consider a biclique $B$ emitted as a maximal biclique by clusters $C(v)$ 
  such that it is not maximal in $G$.
  Then there exists a maximal biclique $M'$ that can be generated by
  extending $B$. However, we know that all edges in $M'$ must also be
  contained in $C(v)$ as all of the extra edges must also be
  in 2--neighborhood of $v$. Thus $B$ will generate $M'$ upon
  extension. This contradicts our assumption that $B$ was
  generated.
}

There are two problems with the basic clustering approach described
above. First is {\em redundant work}. Although each maximal biclique
in $G$ is emitted only once, it may still be generated multiple times,
in different clusters. This redundant work significantly adds to the
runtime of the algorithm.
Second is an {\em uneven distribution of load} among different
subproblems. The load on subproblem $C(v)$ depends on two factors, the
complexity of cluster $C(v)$ (i.e. the number and size of maximal
bicliques within $C(v)$) and the position of $v$ in the total order of
the vertices. The earlier $v$ appears in the total order, the greater
is the likelihood that a maximum biclique in $C(v)$ has $v$ has its
smallest vertex, and hence the greater is the responsibility for
emitting bicliques that are maximal within $C(v)$. Using a
lexicographic ordering of the vertices may lead to a significantly
increased workload for clusters of lower numbered vertices and a
correspondingly low workload for clusters of higher numbered vertices.

\remove{Now if the redundant work is eliminated, we face a second
  issue.  The strategy to eliminate duplicates results in improper
  work distribution among the vertices. The strategy is to make the
  cluster $C(v)$ emit a maximal biclique $M$ if $v$ is the smallest
  vertex in $M$. Smaller vertices in $G$ will be the smallest vertex
  for more maximal bicliques than the larger vertices.  Hence, the
  smaller vertices will have to do more work than the smaller vertices
  resulting in unequal work distribution. As we will see, solving this
  problem also results in significant performance gains over the
  algorithm without Load Balancing.  The various algorithms designed
  for MBE are shown in Table~\ref{table:algo}.  CDFS represents the
  basic Clustering DFS approach. CD0 represents the DFS Algorithm with
  optimizations to remove the redundant work.  Finally algorithms CD1
  and CD2 represent the two load balancing approaches used.  As a
  notation we use CDL to denote both CD1 and CD2.  }

\begin{table*}
\caption{Different versions of Parallel Algorithms based on Depth First Search (DFS)}
\label{table:algo}
\centering
\begin{tabular}{c c} 
Label & Algorithm \\ [0.5ex]
\hline\hline 
CDFS & Clustering based on Depth First Search (DFS) \\
CD0 & CDFS + Reducing Redundant Work, without Load Balancing \\
CD1 & CDFS + Reducing Redundant Work + Load Balancing using Degree \\
CD2 & CDFS + Reducing Redundant Work + Load Balancing using Size of 2-neighborhood \\
\hline
\end{tabular}
\end{table*}

\remove{
\begin{figure*}
\begin{center}

\begin{tikzpicture}[node distance = 2cm, auto]

\node [block] (GALM) {Generate Adjacency List (Map)};
\node [block, right of=GALM] (GALR) {Generate Adjacency List (Reduce)};
\node [block, right of=GALR] (NCM) {Generate Two Neighborhood (Map)};
\node [block, right of=NCM] (NCR) {Generate Two Neighborhood and run sequential algorithm (Reduce)};

\path [line] (GALM) -- (GALR);
\path [line] (GALR) -- (NCM);
\path [line] (NCM) -- (NCR);

\end{tikzpicture}

\caption{Schematic for Clustering Approach} \label{schematic:clustering}
\end{center}
\end{figure*}
}

\subsection{Reducing Redundant Work}

In order to reduce redundant work done at different clusters, we
modify the sequential DFS algorithm for MBE that is executed at each
reducer. We first observe that in cluster $C(v)$, the only maximal
bicliques that matter are those with $v$ as the smallest vertex; the
remaining maximal bicliques in $C(v)$ will not be emitted by this
reducer, and need not be searched for here. We use this to prune
the search space of the sequential DFS algorithm used at the reducer.

All search paths in the algorithm which lead to a maximal biclique
having a vertex less than $v$ can be pruned away. Hence, before
starting the DFS, we prune away all vertices in the Tail set that are
less than $v$, as described in Algorithm~\ref{NCR_DFS}. Also, in DFS
Algorithm~\ref{CPA}, we prune the search path in Line 12 if the
generated neighborhood contains a vertex less than $v$ -- maximal
bicliques along this search path will not have $v$ as the smallest
vertex. Finally in Line 19 of Algorithm~\ref{CPA}, we emit a maximal
biclique only if the smallest vertex is the same as the key of the
reducer in Algorithm~\ref{NCR_DFS}.

The above algorithm, the ``optimized DFS clustering algorithm'', or
``CD0'' for short, is described in Algorithm~\ref{algo:CD0}. This
takes two rounds of MapReduce.  The first round, described in
Algorithms~\ref{GALM} (map) and~\ref{GALR} (reduce), is responsible
for generating the 1-neighborhood for each vertex. The second round,
described in Algorithms~\ref{NCM} (map) and~\ref{NCR_DFS} (reduce)
first constructs the clusters $C(v)$ and runs the optimized sequential
pruning algorithm at the reducer. 
Note that Algorithm~\ref{NCR_DFS} passes the size threshold $s$
while calling the optimized DFS Algorithm~\ref{CPA}. 
The size threshold is an user input and can be passed on to 
Reducer (Algorithm~\ref{NCR_DFS}) by using the \emph{Configuration }
parameters of Hadoop. Like the sequential algorithm,
this parameter can be set to $1$ to enumerate all maximal bicliques 
and to a larger value to  enumerate only large maximal bicliques.

\begin{algorithm} 
\caption{Algorithm CD0} 
\KwIn{Edge List of $G=(V,E)$}
\label{algo:CD0}
Generate Adjacency List (Map) -- Algorithm~\ref{GALM} \;
Generate Adjacency List (Reduce) -- Algorithm~\ref{GALR} \;
Create Two Neighborhood (Map) -- Algorithm~\ref{NCM} \;
Create Two Neighborhood (Reduce) -- DFS -- Algorithm~\ref{NCR_DFS} \;
\end{algorithm} 

\begin{algorithm} 
\caption{Generate Adjacency List -- Map} 
\KwIn{Edge $(x,y)$}
\label{GALM}
Emit ($key \leftarrow x$,$value \leftarrow y$) \;
Emit ($key \leftarrow y$,$value \leftarrow x$) \;
\end{algorithm} 

\begin{algorithm} 
\caption{Generate Adjacency List -- Reduce}
\KwIn{$key \leftarrow v$,$value \leftarrow $\{Neighbors of $v$\}}
\label{GALR}
$neighborhood \leftarrow \varnothing$ \;
\ForAll{$val \in value$}
{
	$neighborhood \leftarrow neighborhood \cup val$ \;
}
$N \leftarrow \langle v,neighborhood \rangle $  \;
Emit ($key \leftarrow \varnothing$,$value \leftarrow N$) \;
\end{algorithm}

\begin{algorithm} 
\caption{Create Two Neighborhood -- Map}
\KwIn{$N \leftarrow \langle v,neighborhood \rangle$}
\label{NCM}
Emit ($key \leftarrow v$,$value \leftarrow N$) \;
\ForAll{$y  \in neighborhood$}
{
	Emit ($key \leftarrow y$,$value \leftarrow N$) \;
}
\end{algorithm}

\remove{
\begin{algorithm} 
\caption{Create Two Neighborhood (Consensus) -- Reduce} 
\KwIn{$key \leftarrow vertex$, $value \leftarrow$ \{Collection of neighborhoods\} )}
\label{NCR_Cons}
$G'=(V',E') \leftarrow $ Induced 2--neighborhood sub--graph created from $value$ \;
$OUT \leftarrow $ Output set of maximal bicliques from sequential Algorithm \;
\ForAll{$biclique$ contained in set $OUTPUT$}
{
	$min \leftarrow $ minimum vertex in $biclique$ \;
	\If{$key$ = $min$}  
	{  
		Emit ($key \leftarrow \varnothing$,$value \leftarrow biclique$); \tcp{Duplicate elimination} 
	}
}
\end{algorithm}
}

\begin{algorithm}
\caption{Create Two Neighborhood (DFS) -- Reduce} 
\KwIn{$key \leftarrow v$, $value \leftarrow$ \{2--neighborhood of $v$\} }
\label{NCR_DFS}
$G'=(V',E') \leftarrow $ Induced subgraph on $\eta^2(v)$ \;
$X \leftarrow \varnothing$ \;
$T \leftarrow V'$ \;
\ForAll{vertex $t \in T$}
{
	\If{$t < key$}
	{
			$T \leftarrow T \setminus \{ t \}$ \;
	}
}
CD0\_Seq($G'$, $X$, $T$, $key$, $s$) \;
\end{algorithm}

\begin{algorithm} 
\caption{Optimized DFS -- CD0\_Seq} 
\KwIn{$G'$,$X$,$T$,$key$,$s$}
\label{CPA}

	\ForAll{vertex $v \in T$}
	{
		\If{$\left | \Gamma(X \cup \{ v \}) \right | < s$}
		{
			$T \leftarrow T \setminus \{ v \}$ \;
		}
	}
	
	\If{$\left | X \right | + \left | T \right | < s$}
	{
		\Return
	}
	
	Sort vertices in $T$ as per ascending order of $\left | \Gamma(X \cup \{ v \}) \right |$ \;
	
	\ForAll{vertex $v \in T$}
	{
		$T \leftarrow T \setminus \{ v \}$ \;
		
		\If{$\left | X \cup \{ v \} \right | + \left | T \right | \ge s$}
		{

			$N  \leftarrow \Gamma( X \cup \{ v \} )$ \;
			$Y \leftarrow \Gamma( N )$ \; 

			\If{ $Y$ contains vertices smaller than $key$}
			{
				continue \;
			}

			Biclique $B \leftarrow \left < Y, N \right >$ \;

			\If{$ ( Y \setminus ( X \cup \{ v \} ) ) \subseteq T $}
			{
				\If{$\left | Y \right | \ge s$}
				{
					$v_s \leftarrow$ Smallest vertex in $B$ \;
					\If{$v_s = key$}  
					{
                                               \tcp{Maximal biclique found}
						Emit ($key \leftarrow \varnothing$,$value \leftarrow B$) \;  
					}
				}	

				CD0\_Seq($G'$, $Y$, $T \setminus Y$, $key$, $s$) \;
			}

		}
	}

\end{algorithm} 

\remove{
We eliminate duplicate work as follows.
Without loss of generality, consider a cluster $C(v)$.
Only those maximal bicliques having minimum vertex as $v$ will be 
enumerated from $C(v)$.
This means that all the search paths in the parallel DFS Algorithm~\ref{CPA}
which lead to maximal bicliques having smallest vertex less than $v$
in the lexicographical ordering can be pruned to avoid the redundant work.
This would ensure that all the search paths lead to maximal bicliques
having smallest vertex $v$ and thus should be enumerated.
Hence before starting the DFS, we prune out all the vertices in the Tail set
that are less than $v$. This is done in Algorithm~\ref{NCR_DFS}.
Also, in the parallel DFS Algorithm~\ref{CPA}, we prune the search path
in Line 12 if the generated neighborhood information contains a vertex
less than $v$. This is because, the generated maximal biclique from this
neighborhood information will have smallest vertex less than $v$ and thus
will not be enumerated by the algorithm while processing $C(V)$.
Finally in Line 19 of Algorithm~\ref{CPA}, we emit a maximal biclique
only if the smallest vertex is the same as the key of the reducer in Algorithm~\ref{NCR_DFS}.
}




\begin{lemma}
\label{lemma:DFS_Opt}
No maximal biclique in G is emitted by more than one reducer in  
Algorithm~\ref{CPA}.
\end{lemma}

\begin{proof}
Without the loss in generality, consider any maximal biclique 
$M~=~<L,R>$. Let $a \in \{ L \cup R \}$ be the smallest vertex
in $\{ L \cup R \}$. Consider the reducer with $key = a$.
In Line 20 of Algorithm~\ref{CPA}, a maximal biclique
is emitted only if the condition in line 18 is satisfied.
This condition is satisfied by the reducer with $key = a$.
However, this condition is not satisfied for any reducer
such that $key \neq a$. Thus maximal biclique $M$ is emitted
only for the reducer with $key = a$. 
\end{proof}

\begin{lemma}
\label{lemma:DFS}
Algorithm~\ref{algo:CD0} generates all maximal bicliques in a graph.
\end{lemma}

\begin{proof}
  The correctness of this Lemma can be proved from 
  Lemmas~\ref{lemma:part}~and~\ref{lemma:DFS_Opt}.
  Algorithm ~\ref{algo:CD0} generates the 2-neighborhood induced sub--graph
  of each vertex in $G$. It then runs the Sequential DFS algorithm with the
  optimizations explained above.

  The correctness relies on the following two observations: Firstly, 
  from Lemma~\ref{lemma:DFS_Opt}, a maximal biclique is emitted
  from a reducer only if the smallest vertex in the biclique is same 
  as the reducer key. Secondly, no vertex is ever removed from the set $X$. 
  The set $X$ thus always
  grows in size and never gets smaller in the course of the
  depth-first search. This is because the set $Y$ is generated from
  set $X$ in line 11 of Algorithm~\ref{CPA} and the set $Y$ is passed
  as the new set $X$ for the next level of recursion.  The set $Y$ is
  generated from the set $X$ by taking the neighborhood of
  neighborhood of set $X$.  $\eta(X)$ contains the set of all vertices
  connected to all vertices in $X$.  Then $\Gamma(\Gamma(X))$ contains all
  vertices connected to all vertices in $\eta(X)$.  This must include
  $X$. Hence $Y \supseteq X$.

  From the above two observations we can prove the Lemma.  Since we
  emit only those maximal bicliques for which the smallest vertices is
  the same as the reducer key $k$, we do not need to search the paths
  that produce maximal bicliques with smallest vertex less than $k$.
  Also, since no vertex is ever removed from the set $X$ through the
  recursion path, we can be sure that at no point in the execution of
  the algorithm we will have $v \in X$ such that $v < k$.  Now the set
  $T$ can be considered as the candidate set as we always add elements
  to set $X$ from set $T$.  Thus in Algorithm~\ref{NCR_DFS} we remove
  all vertices from the set $T$ that are less than $k$.  Further in
  Algorithm~\ref{CPA}, if we generate a maximal biclique in Line 12
  with minimum vertex less than $k$ then we prune the search tree
  through that path as all further maximal bicliques found in that
  search path will contain that vertex less than $k$. 
\end{proof}


\subsection{Load Balancing}

\begin{algorithm} 
\caption{Algorithms CD1 and CD2} 
\KwIn{Edge List of $G=(V,E)$}
\label{algo:CDL}
Generate Adjacency List (Map) -- Algorithm~\ref{GALM} \;
Generate Adjacency List (Reduce) -- Algorithm~\ref{GALR} \;
Send vertex property (Map) -- Algorithm~\ref{NCM} \;
Send vertex property (Reduce) -- Algorithm~\ref{NCR_DFS_New} \;
Create Two Neighborhood (Map) -- Algorithm~\ref{NCM_DFS_Three} \;
Create Two Neighborhood (Reduce) -- Algorithm~\ref{NCR_DFS_Three} \;
\end{algorithm} 

\begin{algorithm}
\caption{Send vertex property -- Reduce }
\KwIn{$key \leftarrow v$, $value \leftarrow$ \{2--neighborhoods of $v$\} }
\label{NCR_DFS_New}
$S \leftarrow$ 2--neighbors of $v$ \;
$neighborhood \leftarrow$ Compute neighborhood of $v$ from $S$ \;
$N \leftarrow \langle v,neighborhood \rangle $  \;
\tcp{Need to pass neighborhood for Round 3} 
Emit ($key \leftarrow \varnothing$,$value \leftarrow N$) \;
\tcp{Need to send vertex property to all 2--neighbors} 
$p \leftarrow$ Value of vertex property of $v$ from $S$ \;
\ForAll{vertices $s \in S$}
{
	Emit($key \leftarrow \varnothing$,$value \leftarrow \lbrack s,v,p \rbrack$) \;
}
\end{algorithm}

\begin{algorithm} 
\caption{Create Two Neighborhood with vertex property -- Map} 
\KwIn{$N \leftarrow \langle v,neighborhood \rangle$ OR $\lbrack s,v,p \rbrack$}
\label{NCM_DFS_Three}
\If {$Input = N$}
{
        Emit ($key \leftarrow v$,$value \leftarrow N$) \;
        \ForAll{$y  \in neighborhood$}
        {
	        Emit ($key \leftarrow y$,$value \leftarrow N$) \;
        }
}
\ElseIf {$Value = \lbrack s,v,p \rbrack$}
{
        Emit ($key \leftarrow s$,$value \leftarrow \lbrack v,p \rbrack$) \;
}
\end{algorithm} 

\begin{algorithm}
\caption{Create Two Neighborhood with vertex property (DFS) -- Reduce }
\KwIn{$key \leftarrow v$, $value \leftarrow$ \{$\eta^2(v)$ along with vertex properties\}}
\label{NCR_DFS_Three}
$G'=(V',E') \leftarrow $ Induced subgraph on $\eta^2(v)$ \;
$Map \leftarrow $ HashMap of vertex and vertex property created from $value$ required to compute the new ordering \;
$X \leftarrow \varnothing$ \;
$T \leftarrow V'$ \;
\ForAll{vertex $t \in T$}
{
	\If{$t < key$ in the new ordering}
	{
			$T \leftarrow T \setminus \{ t \}$ \;
	}
}
CDL\_Seq($G'$, $X$, $T$, $key$,$Map$, $s$) \;
\end{algorithm}

\begin{algorithm} 
\caption{Load Balanced DFS -- CDL\_Seq} 
\KwIn{$G'$,$X$,$T$,$key$,$Map$,$s$}
\label{CPA_New}

	\ForAll{vertex $v \in T$}
	{
		\If{$\left | \Gamma(X \cup \{ v \}) \right | < s$}
		{
			$T \leftarrow T \setminus \{ v \}$ \;
		}
	}
	
	\If{$\left | X \right | + \left | T \right | < s$}
	{
		\Return
	}
	
	Sort vertices in $T$ as per ascending order of $\left | \Gamma(X \cup \{ v \}) \right |$ \;
	
	\ForAll{vertex $v \in T$}
	{
		$T \leftarrow T \setminus \{ v \}$ \;
		
		\If{$\left | X \cup \{ v \} \right | + \left | T \right | \ge s$}
		{

			$N  \leftarrow \Gamma( X \cup \{ v \} )$ \;
			$Y \leftarrow \Gamma( N )$ \; 

			\If{ $Y$ contains vertices smaller than $key$ in the new ordering}
			{
				continue \;
			}

			Biclique $B \leftarrow \left < Y, N \right >$ \;

			\If{$ ( Y \setminus ( X \cup \{ v \} ) ) \subseteq T $}
			{
				\If{$\left | Y \right | \ge s$}
				{
					$v_s \leftarrow$ Smallest vertex in $B$ in the new ordering \;
					\If{$v_s = key$}
					{
                                               \tcp{Maximal biclique found}
						Emit ($key \leftarrow \varnothing$,$value \leftarrow B$) \;  
					}
				}	

				CDL\_Seq($G'$, $Y$, $T \setminus Y$, $key$,$Map$, $s$) \;
			}

		}
	}

\end{algorithm} 

In Algorithm~\ref{algo:CD0}, lexicographical ordering was used to
order the vertices, which is agnostic of the properties of the cluster
$C(v)$.  The way the optimized DFS works (Algorithm~\ref{CPA}), a
reducer processing a vertex that is earlier in the total order is
responsible for emitting more of the maximal bicliques within its
cluster.

For improving load balance, we adjust the position of vertex $v$ in
the total order according to the properties of its cluster $C(v)$.
Intuitively, the more complex cluster $C(v)$ is (i.e. more and larger
the maximal bicliques), the higher should be position of $v$ in the
total order, so that the burden on the reducer handling $C(v)$ is
reduced. While it is hard to compute (or estimate) the number of
maximal bicliques in $C(v)$, we consider two properties of vertex $v$
that are simpler to estimate, to determine the relative ordering of
$v$ in the total order: (1)~Size of 1-neighborhood of $v$ (Degree),
and (2)~Size of 2-neighborhood of $v$

In case of a tie, the vertex ID is used as a tiebreaker. These
approaches were considered since vertices with higher degree are
potentially part of a denser part of the graph and are contained
within a greater number of maximal bicliques. The size of the
2-neighborhood gives the number of vertices in $C(v)$ and may provide
a better estimate of the complexity of handling $C(v)$, but this is
more expensive to compute than the size of the 1-neighborhood of the
vertex.

\remove{
Similarly vertices with
a larger 2--neighborhood will be part of more maximal bicliques.  The
strategy is as follows: the more number of maximal bicliques a vertex
is a member of, the less amount of work should it do.  {\color{blue}
  If possible, need to check if there are some theoretical results
  supporting this claim}.
}

The discussion below is generic and holds for both approaches to load
balancing. To run the load balanced version of DFS, the reducer
running the sequential algorithm must now have the following
information for the vertex (key of the reducer) : (1)~2-neighborhood
induced subgraph, and (2)~vertex property for every vertex in the
2-neighborhood induced subgraph, where ``vertex property'' is the
property used to determine the total order, be it the degree of the
vertex or the size of the 2-neighborhood.  The second piece of
information is required to compute the new vertex ordering. However,
the reducer of the second round does not have this information for
every vertex in $C(v)$, and a third round of MapReduce is needed to
disseminate this information among all reducers.  Further details are
described in Algorithm~\ref{algo:CDL}.  The DFS sequential algorithm
for load balancing, described in Algorithm~\ref{NCR_DFS_Three}, is the
same as the optimized DFS sequential algorithm~\ref{NCR_DFS}, except
that it orders using the vertex property (ties broken by IDs) rather
than the simple lexicographic ordering.

\subsection{Communication Complexity}

We consider the communication complexity of Algorithms CD0, CD1 and CD2.
For input graph $G = \left ( V, E \right )$,
we know that $n = \left | V \right |$ and $m = \left | E \right |$.
Let us assume $\Delta$ to be the largest degree and $\bar{d}$ 
to be the average degree of vertices, where $\bar{d} = m / n$.
Also, let $\beta$ be the Output Size. 

\begin{definition}
\label{def:comm_complx}
Communication complexity of a MapReduce Algorithm $\mathcal{A}$
for Round $r$ is denoted by $\mathcal{C}^r_\mathcal{A}$ and is defined as 
the sum of the total number of bytes emitted by all Mappers and the total 
number of bytes emitted by all the Reducers. 
We consider the output size for Reducers contributing to $\mathcal{C}^r_\mathcal{A}$
as  each Reducer writes into the distributed file system incurring communication.
\end{definition}

\begin{definition}
\label{def:total_comm_complx}
Let $\mathcal{C}_\mathcal{A}$ denote the total communication
complexity for a MapReduce Algorithm $\mathcal{A}$ having $R$ rounds.
We define
$\mathcal{C}_\mathcal{A}$ $=$ $\sum\limits_{r=1}^R \mathcal{C}^r_\mathcal{A}$.
\end{definition}

\begin{lemma}
\label{lemma:complx_CD0}
Total communication complexity of Algorithm CD0 is 
$O \left ( m \cdot \Delta +\beta \right )$.
\end{lemma}

\begin{proof}
Algorithm CD0 has two rounds of MapReduce.
For the first round, Algorithm~\ref{GALM}, which is the Map method emits
each edge twice, resulting in a communication complexity of $O \left ( m \right )$.
Similarly, Algorithm~\ref{GALR}, which is the reducer emits each adjacency
list once. This also results in a communication complexity of $O \left ( m \right )$.
Hence total communication complexity of the first round is $O \left ( m \right )$.

Now let us consider the second round of MapReduce. 
The total communication between the Map and Reduce methods
(Algorithms~\ref{NCM}~and~\ref{NCR_DFS} respectively) can be
computed by analyzing how much data is received by all Reducers.
Each reducer receives the adjacency list of all the neighbors of the key.
Let $d_i$ be the degree of vertex $v_i$, for $v_i \in V$, $i = 1,..,n$.
Total communication is thus 
$\sum\limits_{i=1}^n {d_i}^2 $
$= {d_1}^2 + {d_2}^2 + \cdot \cdot \cdot {d_n}^2$.
This is 
$O \left (  \left ( d_1 + d_2 + \cdot \cdot \cdot + d_n \right ) \cdot \Delta \right ) $
$= O \left (  n \cdot \bar{d} \cdot \Delta \right )$.
Since $\bar{d} = m / n$, total communication becomes 
$O \left ( m \cdot \Delta \right )$.
The output from the final Reducer (Algorithm~\ref{NCR_DFS}) is the collection of all
maximal bicliques and hence the resulting communication cost is $O \left ( \beta \right )$.
Combining two rounds, total communication complexity becomes
$ O \left ( m + m \cdot \Delta +\beta \right )$.
$= O \left ( m \cdot \Delta +\beta \right )$. 
\end{proof}

\begin{lemma}
\label{lemma:complx_CDL}
Total communication complexity of Algorithm CD1 / CD2 is 
$O \left ( m \cdot \Delta +\beta \right )$.
\end{lemma}

\begin{proof}
First, note that both Algorithms CD1 and CD2 have the same communication complexity
and observe that the first round uses the same Map and Reduce methods as CD0. 
Thus communication for Round 1 is $O \left ( m \right )$.
Again, note that Map method for Round 2 is same as CD0 and hence by 
Lemma~\ref{lemma:complx_CD0}, communication for Round 2 is $O \left ( m \cdot \Delta \right )$.

The Reducer (Algorithm~\ref{NCR_DFS_New}) of Round 2 sends the vertex property information to
all its 2--neighbors. Thus every reducer receives information about all of its 2--neighbors.
This makes the total output size of Reducer to be $O \left ( m \cdot \Delta \right )$.
The Map method of Round 3 (Algorithm~\ref{NCM_DFS_Three}) sends out 
the 2--neighborhood information as well as the vertex information to all vertices in 2--neighborhood.
Thus communication cost becomes $O \left ( m \cdot \Delta \right )$.
The Reducer (Algorithm~\ref{NCR_DFS_Three}) emits all maximal bicliques and hence
the resulting communication cost is $O \left ( \beta \right )$.
Thus total communication cost for Algorithms CD1 and CD2 is is 
$O \left ( m \cdot \Delta +\beta \right )$. 
\end{proof}

\subsection{Parallel Consensus}
\label{sec:parallel}

We briefly describe another approach which directly parallelizes the
consensus sequential algorithm of~\cite{AACFHS2004}, in a manner
different from the clustering approach. The motivation for this
approach is as follows.  The clustering approach has the following
potential drawback, it requires each cluster $C(v)$ to have the entire
2-neighborhood of $v$.  For dense graphs, the size of the
2-neighborhood of a vertex can be large, so that the complexity of
each reduce task can be substantial. With the nature of the MapReduce
model, dynamic load balancing among the reducers is not (easily)
possible, so that load balancing will always be an issue for
non-uniform, irregular computations.

Unlike the parallel DFS algorithm which works on subgraphs of $G$, the
consensus algorithm is always directly dealing with bicliques within
graph $G$.  At a high level, it performs two operations repeatedly
(1)~a ``consensus'' operation, which creates new bicliques by
considering the combination of existing bicliques, and (2)~an
``extension'' operation, which extends existing bicliques to form new
maximal bicliques. There is also a need for eliminating duplicates
after each iteration, and also a step needed for detecting
convergence, which happens when the set of maximal bicliques is stable
and does not change further.

We developed a parallel version of each of these operations, by
performing the consensus, extension, duplicate removal, and
convergence test using MapReduce. We omit the details due to lack of
space, but present experimental results from our implementations.

\section{Experimental Results}
\label{sec:results}

\begin{table*}
\caption{Various properties of the input graphs used, 
and runtime of different algorithms using 100 reducers.
DNF means that the algorithm did not finish in 12 hours.
The size threshold was set as $1$ to enumerate all maximal bicliques.
Runtime includes overhead of all MapReduce rounds including graph 
clustering, i.e. formation of 2--neighborhood.}
\tiny
\label{table:inputs}
\centering
\begin{tabular}{c c c c c c c c c c} 
Label & Input Graph & \#vertices & \#edges & \#max--bicliques & Output
Size & CDFS & CD0 & CD1 & CD2 \\ [0.5ex]
\hline\hline 
1 & p2p-Gnutella09 & 8114 & 26013 & 20332 & 407558 & 113 & 92 & 132 & 130 \\
2 & email-EuAll-0.6 & 125551 & 168087 & 292008 & 9161154 & 42023 & 4640 & 683 & 626 \\
3 & com-Amazon & 334863 & 925872 & 706854 & 12739908 & 186 & 113 & 185 & 221 \\
4 & amazon0302 & 262111 & 1234877 & 886776 & 14553776 & 396 & 272 & 151 & 153 \\
5 & com-DBLP-0.6 & 251226 & 419573 & 1875185 & 82814962 & 1659 & 409 & 374 & 478 \\
6 & email-EuAll-0.4 & 175944 & 252075 & 2003426 & 111370926 & DNF & DNF & 6365 & 4154 \\
7 & ego-Facebook-0.6 & 3928 & 35397 & 6597716 & 315555360 & 8657 & 3858 & 1512 & 2943 \\
8 & loc-BrightKite-0.6 & 49142 & 171421 & 10075745 & 777419528 & 28585 & 11451 & 2506 & 2998 \\
9 & web-NotreDame-0.8 & 150615 & 300398 & 19941634 & 942300172 & DNF & DNF & 1688 & 2327 \\
10 & ca-GrQc-0.4 & 5021 & 17409 & 16133368 & 3101214314 & 37279 & 6895 & 5790 & 6374 \\
11 & ER-50K & 50000 & 275659 & 51756 & 1116752 & 96 & 89 & 133 & 136 \\
12 & ER-60K & 60000 & 330015 & 61821 & 1334716 & 98 & 89 & 135 & 135 \\
13 & ER-70K & 70000 & 393410 & 71962 & 1589408 & 98 & 90 & 135 & 132 \\
14 & ER-80K & 80000 & 448289 & 81983 & 1809070 & 102 & 90 & 136 & 134 \\
15 & ER-90K & 90000 & 526943 & 92214 & 2125544 & 109 & 96 & 142 & 140 \\
16 & ER-100K & 100000 & 600038 & 102663 & 2421528 & 114 & 97 & 144 & 143 \\
17 & ER-250K & 250000 & 1562707 & 252996 & 6274864 & 167 & 114 & 165 & 162 \\
18 & ER-500K & 500000 & 3751823 & 506319 & 15057870 & 374 & 167 & 251 & 252 \\
19 & Bipartite-50K-100K & 150000 & 1999002 & 306874 & 9256056 & 873 & 183 & 227 & 253 \\
\hline
\end{tabular}
\end{table*}

\begin{figure*}[!ht]
\begin{center}

$\begin{array}{cc}

\subfloat[ ER-500K and p2p-Gnutella09]{\label{fig:runtime-a}\includegraphics[width=0.45\textwidth]{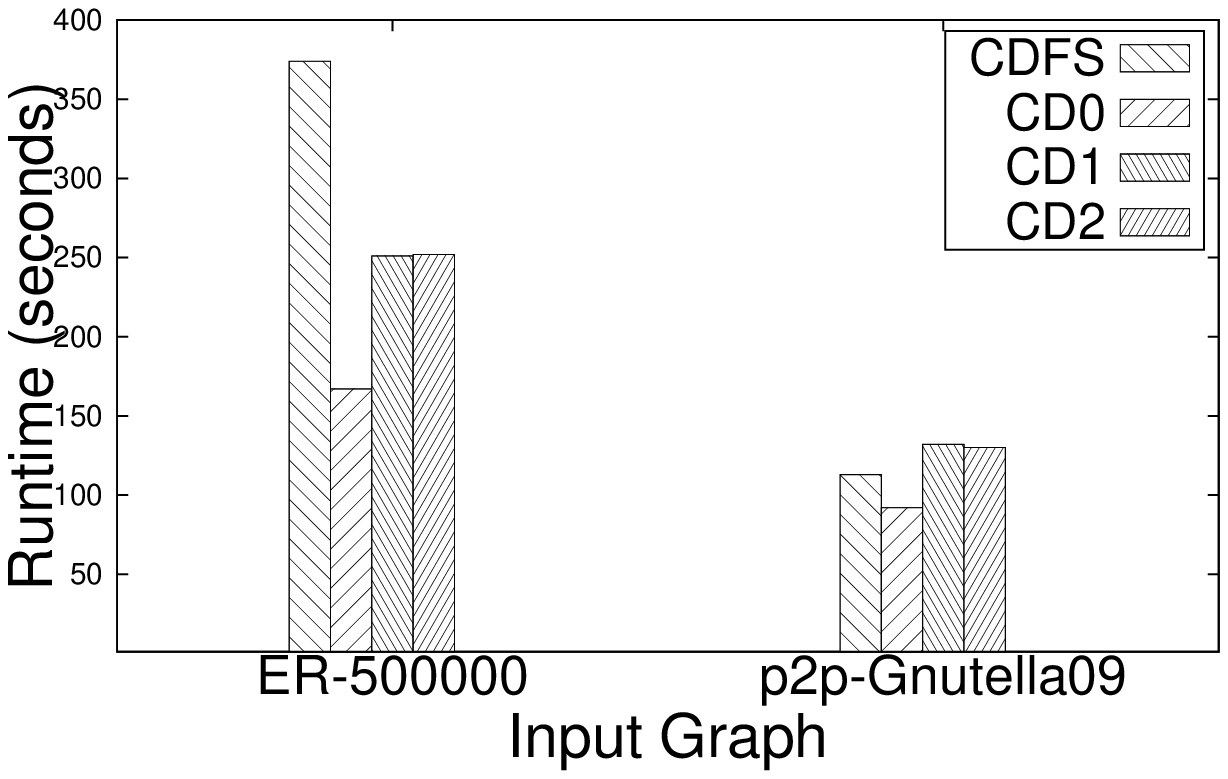}}

\subfloat[ web-NotreDame-0.8 and Bipartite-50K-100K]{\label{fig:runtime-b}\includegraphics[width=0.45\textwidth]{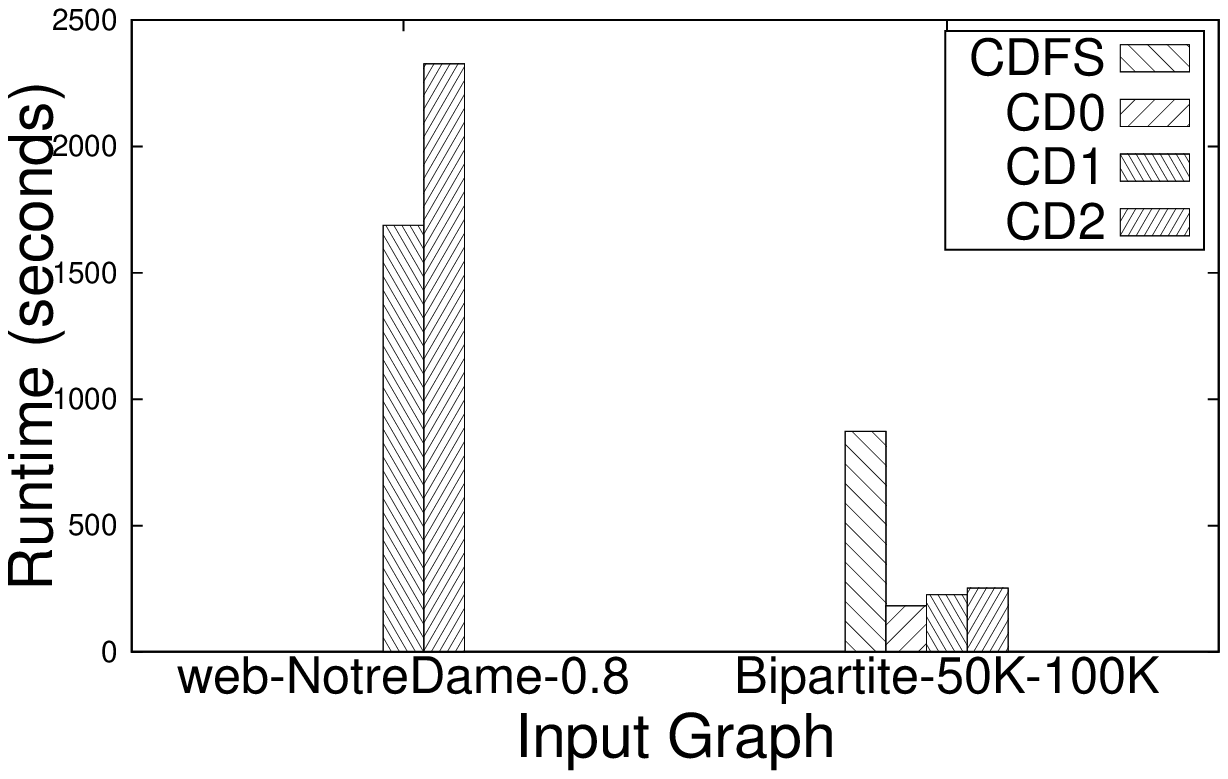}}

\end{array}$

$\begin{array}{cc}

\subfloat[ amazon0302 and com-Amazon]{\label{fig:runtime-c}\includegraphics[width=0.45\textwidth]{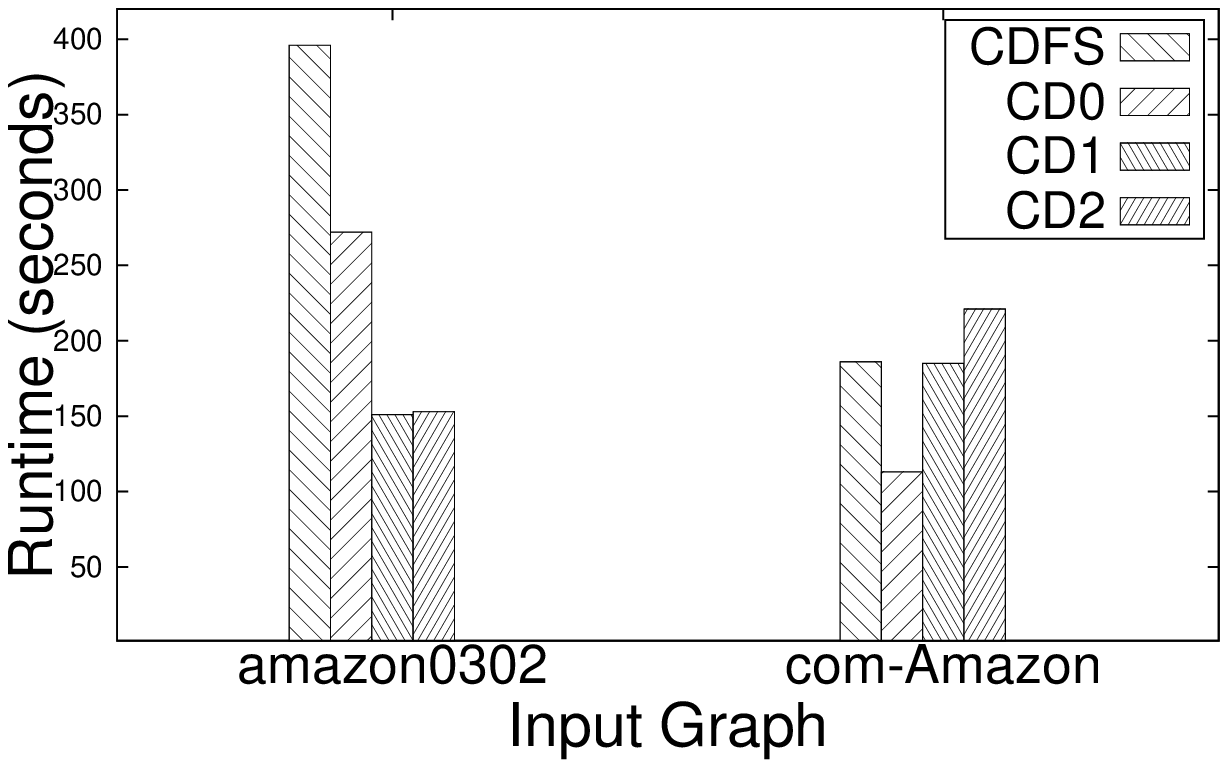}}

\subfloat[ email-EuAll-0.6 and ca-GrQc-0.4]{\label{fig:runtime-d}\includegraphics[width=0.45\textwidth]{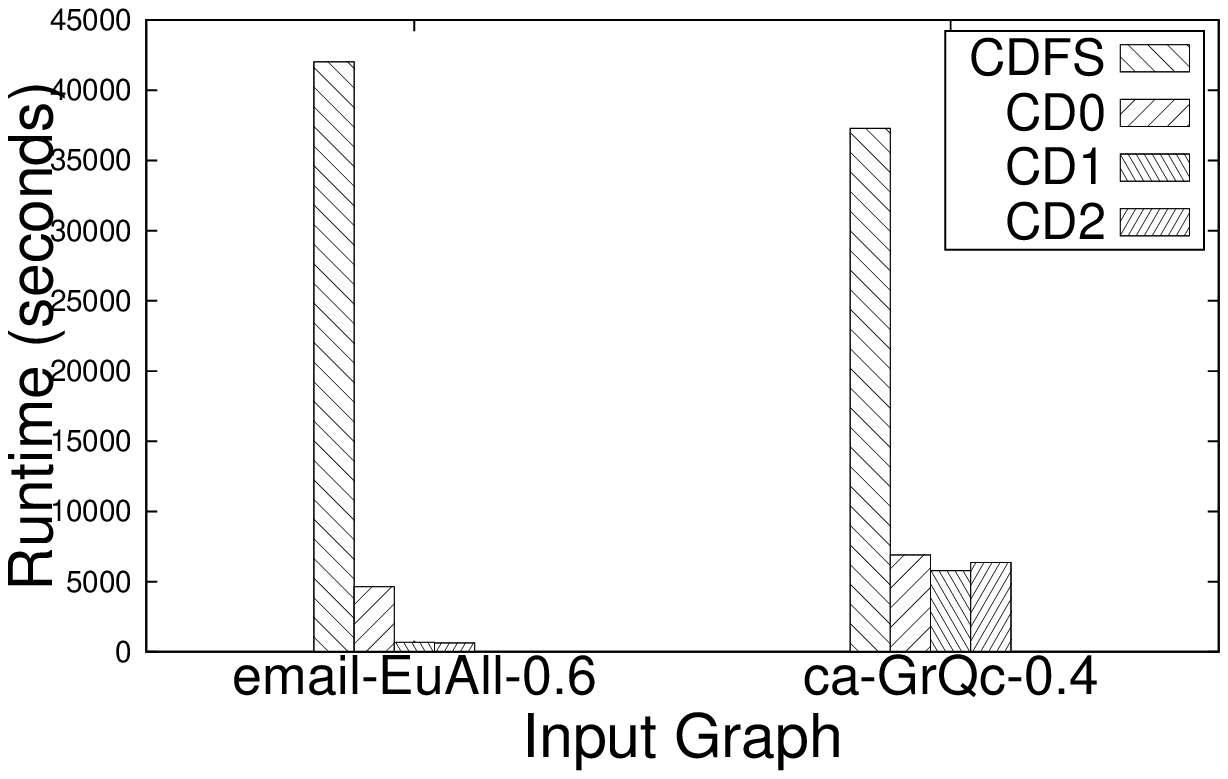}}

\end{array}$

$\begin{array}{cc}

\subfloat[ loc-BrightKite-0.6 and ego-Facebook-0.6]{\label{fig:runtime-e}\includegraphics[width=0.45\textwidth]{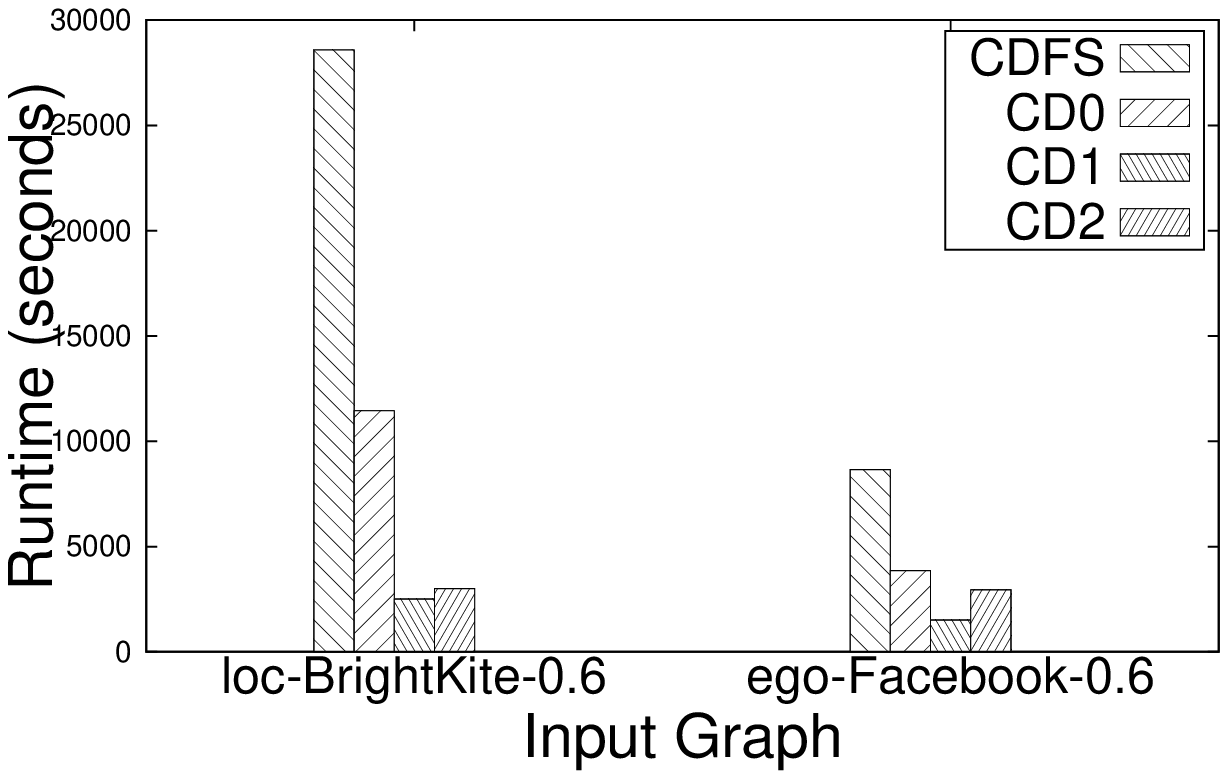}}

\subfloat[ email-EuAll-0.4 and com-DBLP-0.6]{\label{fig:runtime-f}\includegraphics[width=0.45\textwidth]{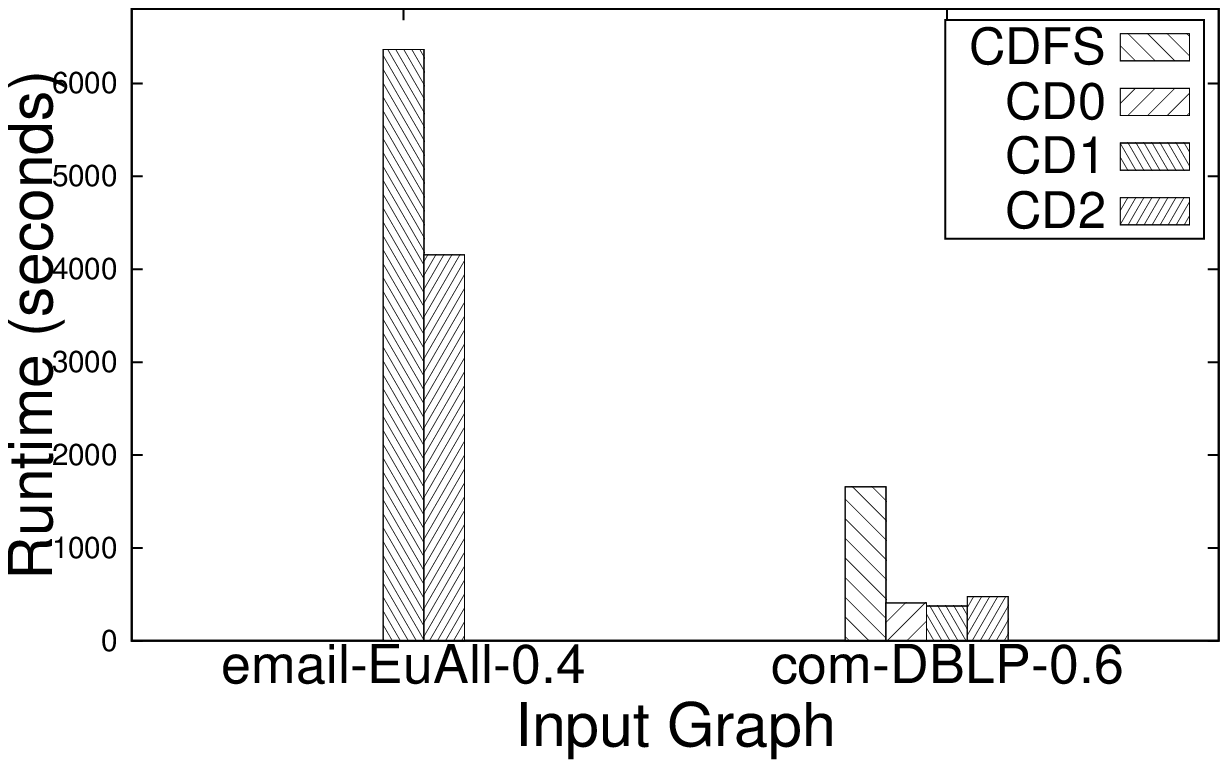}}

\end{array}$

\end{center}

\caption{Runtime of parallel algorithms 
on real and random graphs. If an algorithm failed to complete in 12
hours the result is not shown. All algorithms were run using 100
reducers. Runtime includes overhead of all MapReduce rounds including graph clustering,
i.e. formation of 2--neighborhood.}
\label{fig:runtime}
\end{figure*}

\begin{figure*}[!ht]
\begin{center}
\subfloat[Algorithm CD1]{\label{fig:scaling_LBD}\includegraphics[width=0.45\textwidth]{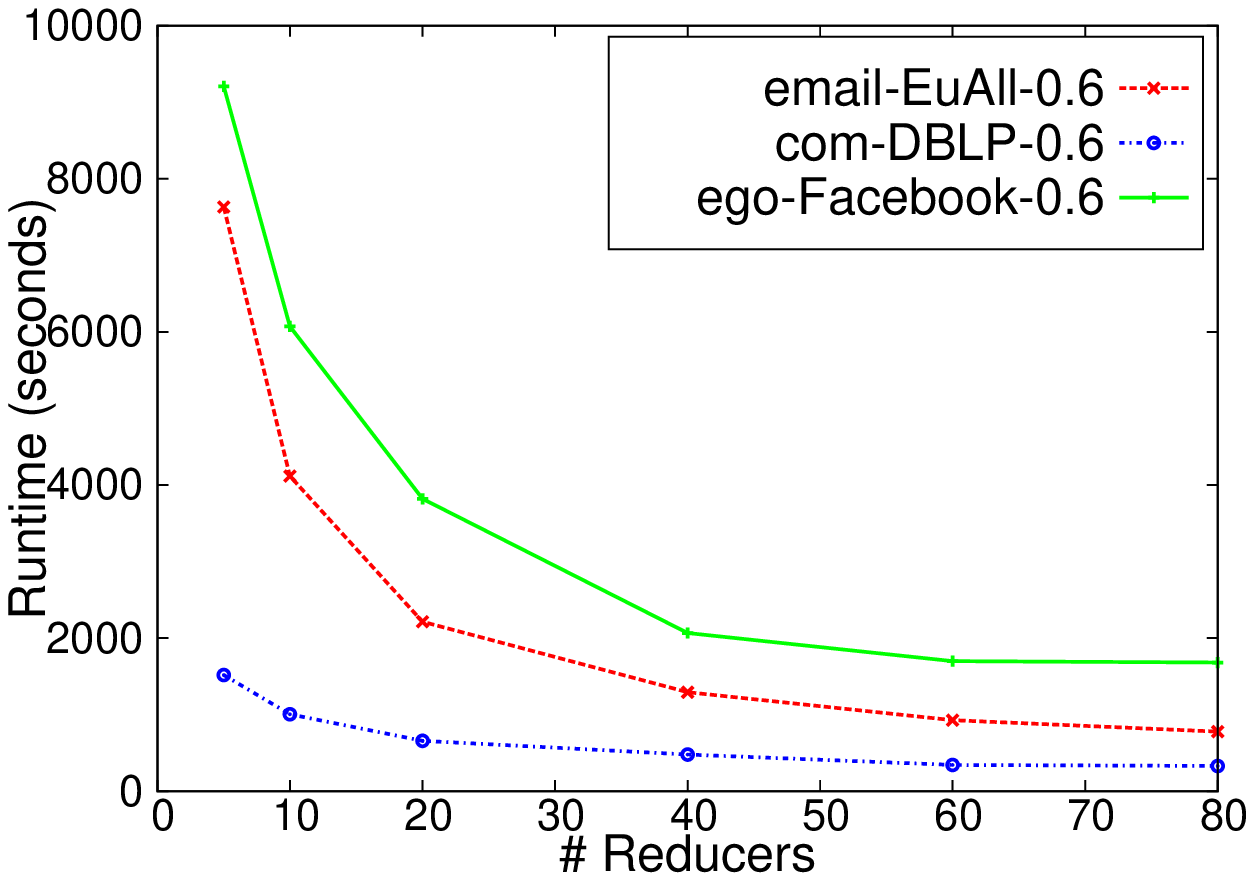}}
\subfloat[Algorithm CD2]{\label{fig:scaling_LBN}\includegraphics[width=0.45\textwidth]{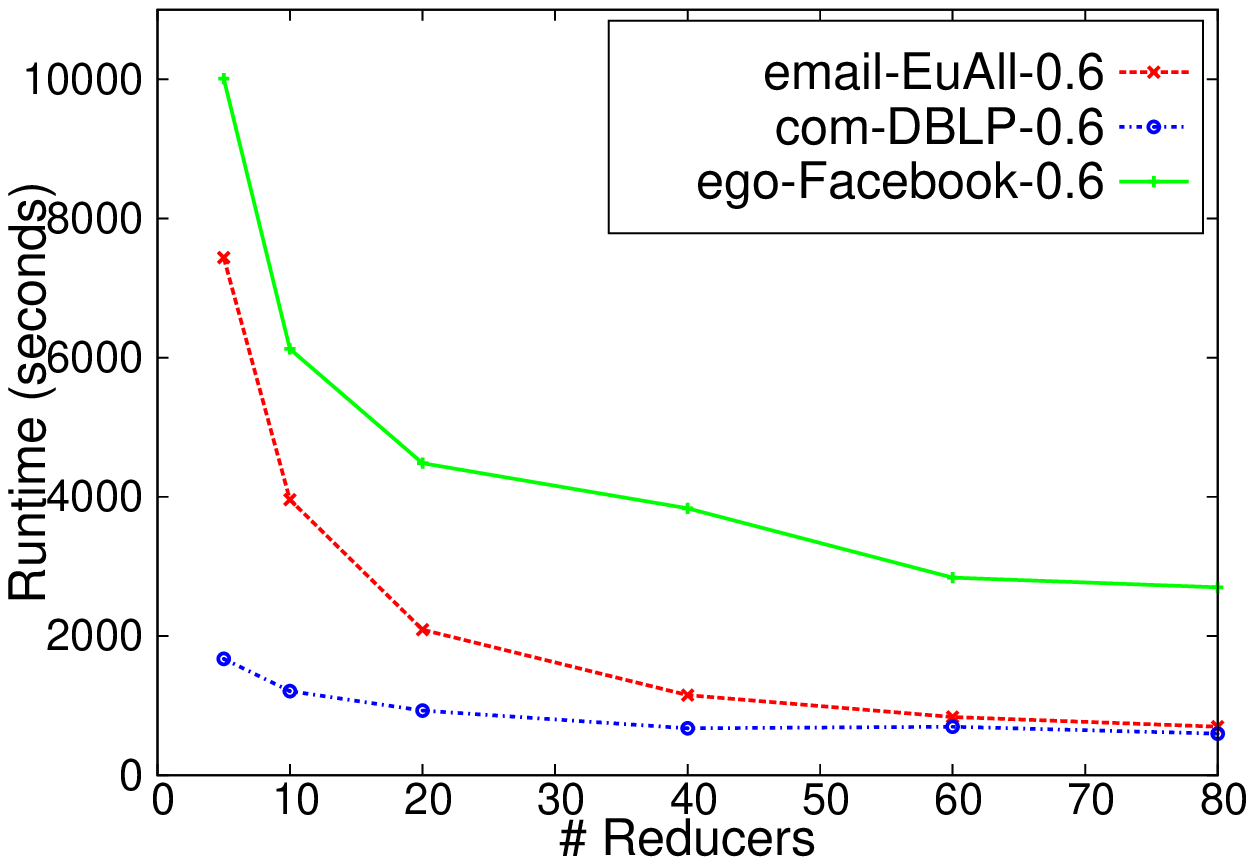}}
\caption{Runtime versus Number of Reducers.} \label{fig:scaling}	
\end{center}
\end{figure*}

\begin{figure*}[!ht]
\begin{center}
\subfloat[Algorithm CD1]{\label{fig:speedup_LBD}\includegraphics[width=0.45\textwidth]{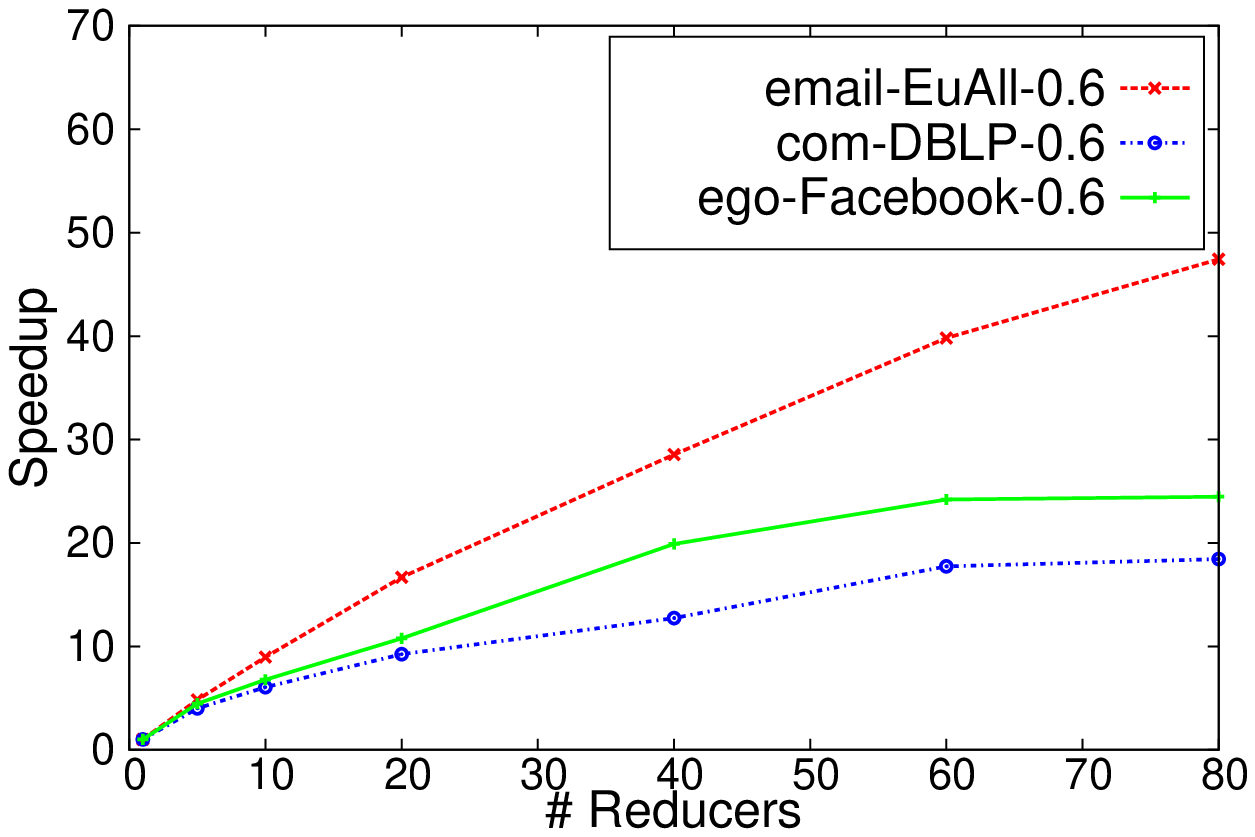}}
\subfloat[Algorithm CD2]{\label{fig:speedup_LBN}\includegraphics[width=0.45\textwidth]{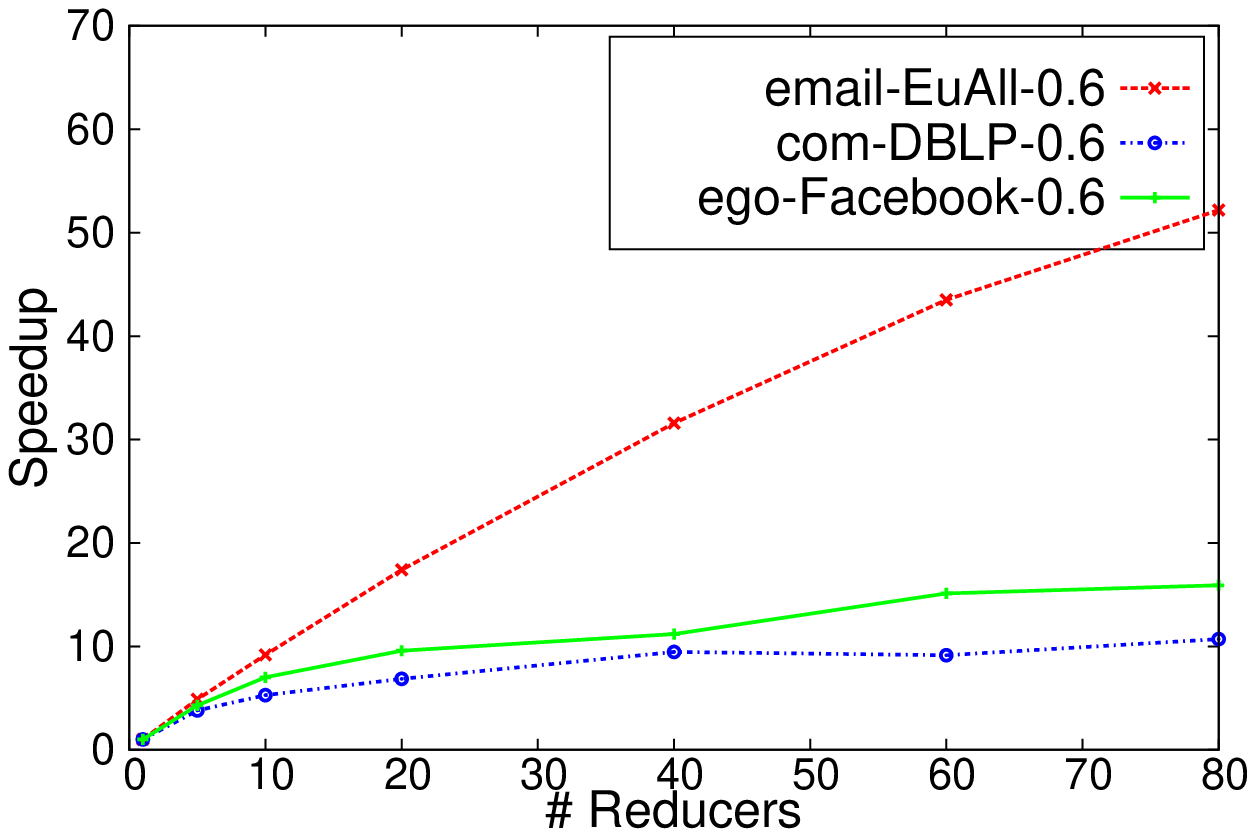}}
\caption{Speedup versus Number of Reducers.} \label{fig:speedup}	
\end{center}
\end{figure*}

\begin{figure*}[!ht]
\begin{center}
\includegraphics[width=0.45\textwidth]{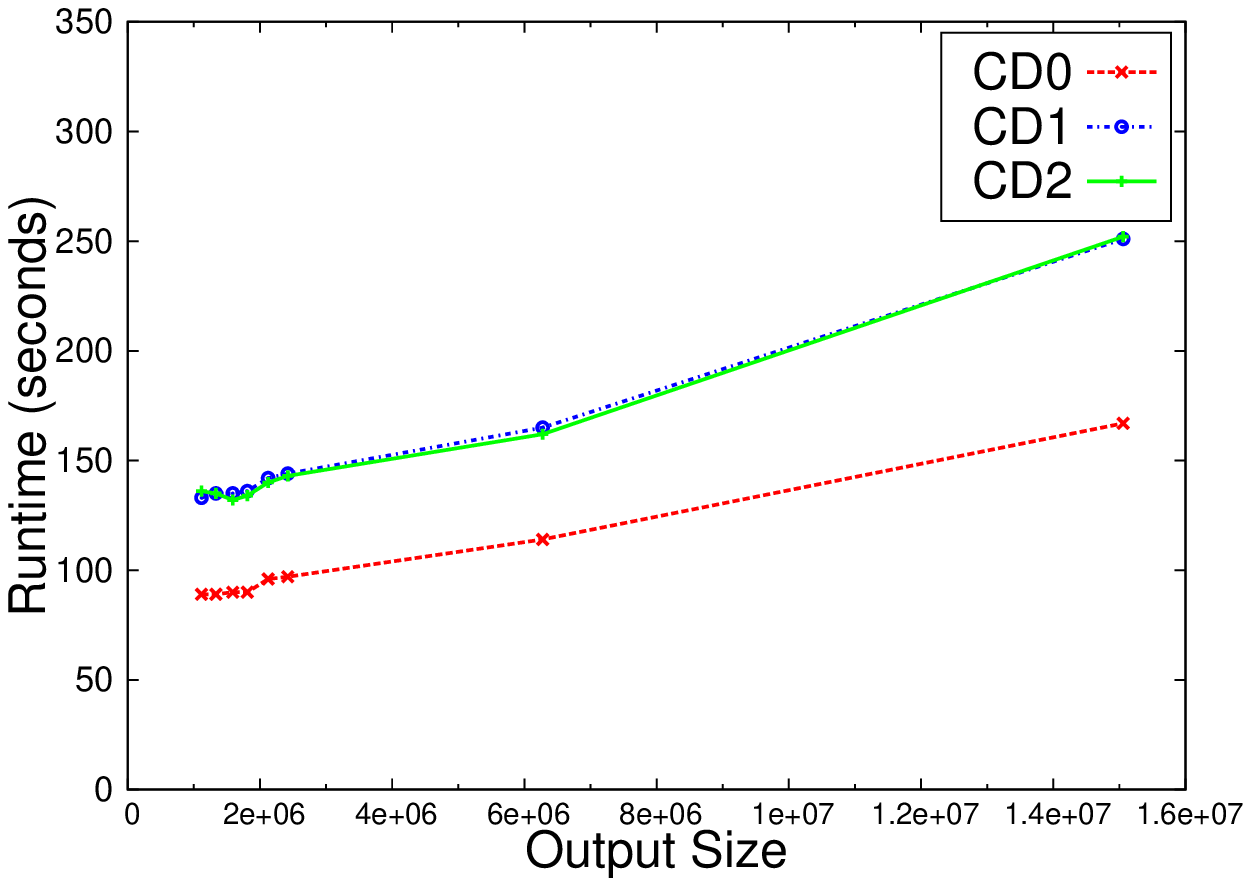}
\caption{Runtime versus Output Size for random graphs. 
All Erdos-Renyi random graphs were used. Output size is defined as the
number of edges summed over all maximal bicliques enumerated.
}
\label{fig:runtime_vs_output}
\end{center}
\end{figure*}

\begin{figure*}[!ht]
\begin{center}
$\begin{array}{cc}

\subfloat[  email-EuAll-0.4 and loc-BrightKite-0.6]{\label{fig:large_maxbicliques-a}\includegraphics[width=0.45\textwidth]{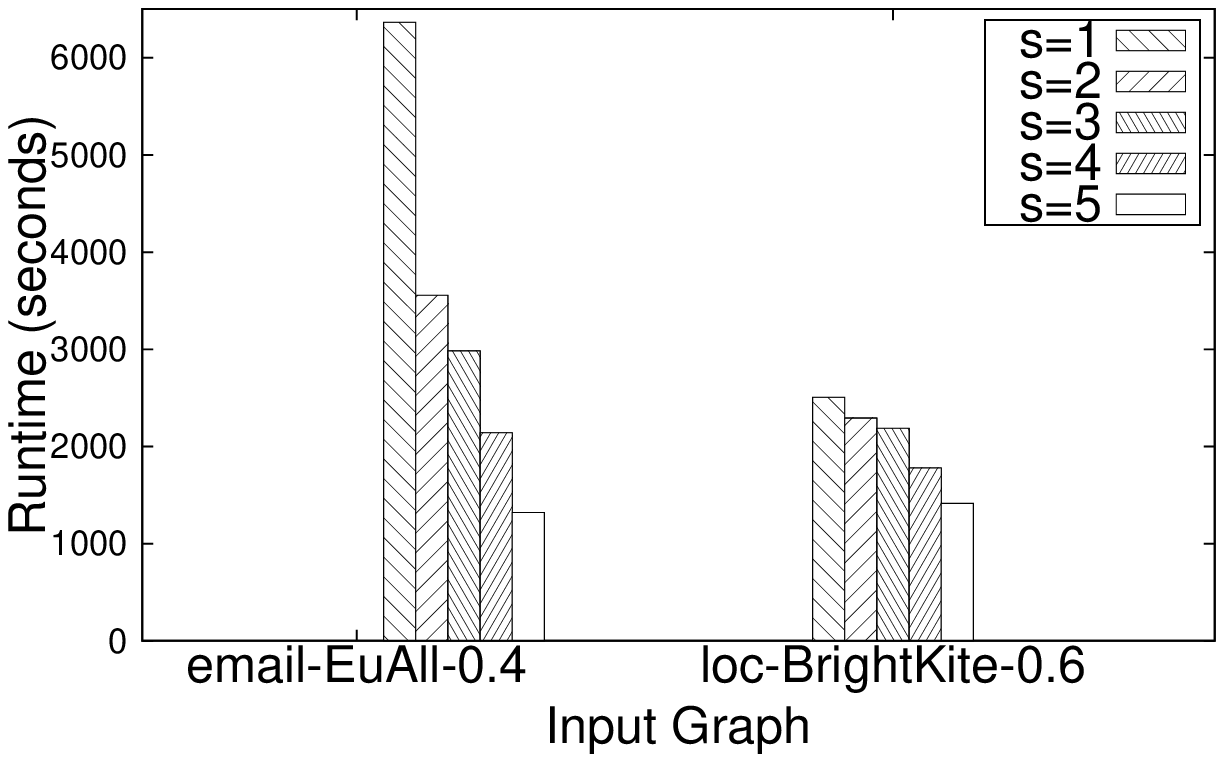}}

\subfloat[ web-NotreDame-0.8 and ego-Facebook-0.6]{\label{fig:large_maxbicliques-b}\includegraphics[width=0.45\textwidth]{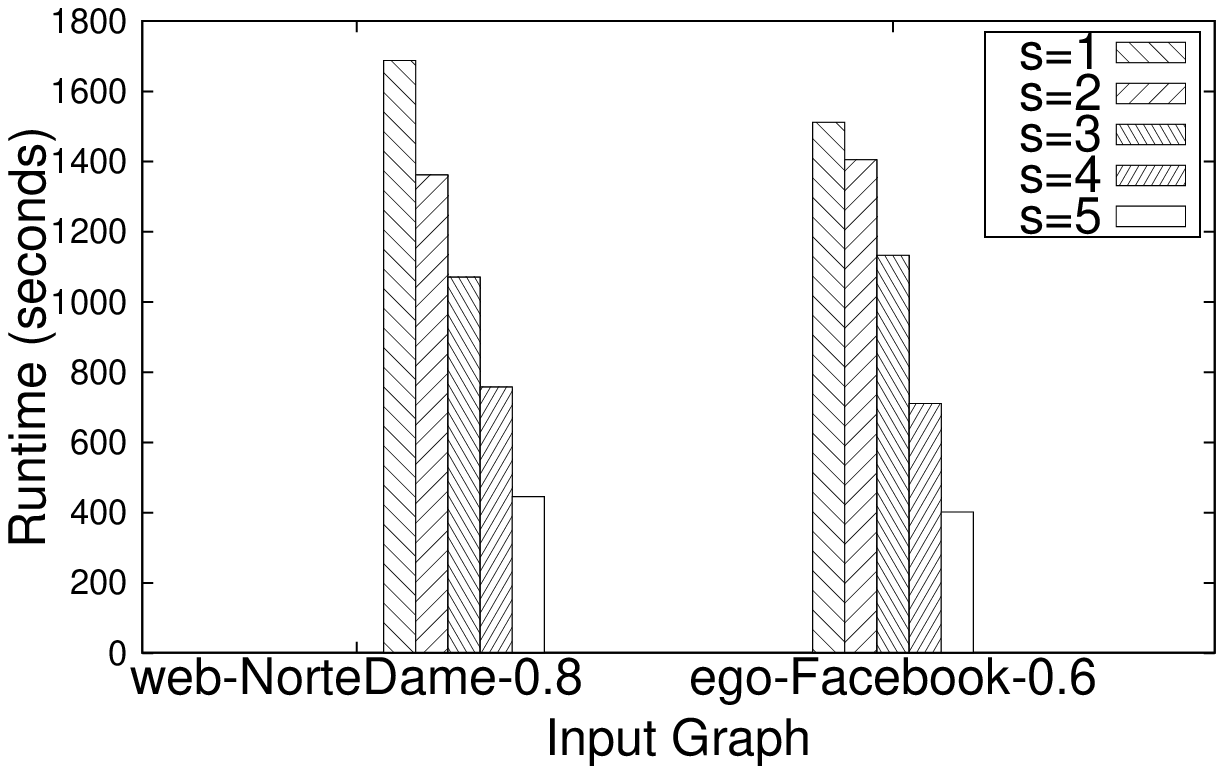}}

\end{array}$
\end{center}
\caption{Runtime vs the size threshold for the emitted maximal bicliques. All experiments were performed using Algorithm CD1 and with 100 reducers.}
\label{fig:large_maxbicliques}
\end{figure*}

We implemented the parallel algorithms on a Hadoop cluster, including
the DFS-based clustering algorithms (CDFS, CD0, CD1, CD2),
consensus-based clustering algorithm, and the parallel consensus
algorithm.
The Hadoop cluster has 28 nodes,
each with a quad-core AMD Opteron processor with 8GB of RAM. All
programs were written using Java and run with 2GB of heap space. The
Java version on all systems was 1.5.0, and the Hadoop version used was
0.20.203.


In addition, we implemented the sequential DFS
algorithm~\cite{LSL2006} (without the optimizations
that we introduced in our work), and the sequential consensus
algorithm (MICA)~\cite{AACFHS2004}. 
The sequential algorithms were not implemented on top of Hadoop
and hence had no associated Hadoop overhead in their runtime. 
But on the real-world graphs that we considered, 
these algorithms did not complete within 12 hours,
except for the p2p-Gnutella09 graph. 

We used both synthetic and real-world graphs for our experiments. A
summary of all the graphs used is shown in Table~\ref{table:inputs}.
The real-world graphs were obtained from the SNAP collection of large
networks~\cite{SNAP} and were drawn from social networks,
collaboration networks, communication networks, product co-purchasing
networks, and internet peer-to-peer networks. Some of the real world
networks were so large and dense that no algorithm was able to process
them. For such graphs, we thinned them down by deleting edges with a
certain probability. This makes the graphs less dense, yet preserves
some of the structure of the real-world graph. We show the edge
deletion probability in the name of the network. For example, graph
``ca-GrQc-0.4'' is obtained from ``ca-GrQc'' by deleting each edge
with probability $0.4$.

Synthetic graphs are either random graphs obtained by the Erdos-Renyi
model~\cite{ER1959}, or random bipartite graphs obtained using a
similar model.
To generate a bipartite graph with $n_1$ and $n_2$
vertices in the two partitions, we randomly assign an edge between
each vertex in the left partition to each vertex in the right
partition. A random Erdos-Renyi graph on $n$ vertices is named
``ER-$<n>$'', and a random bipartite graph with $n_1$ and $n_2$
vertices in the bipartitions is called ``Bipartite-$<n_1>$-$<n_2>$''.

\remove{
generate the general random graphs, for each possible pair of vertices, 
we generated an edge with constant probability $ln(n) / n$,
where $n$ is the number of vertices in the graph. 
For the bipartite random graphs, for each pair of vertices from the \em{left} and \em{right}
set, we generated an edge with probability $ln(n) / n * 5$.
We used a higher probability for the bipartite graphs as for Bipartite graphs many of the possible
edges (between the vertices in the same set) are avoided and hence a higher probability
was used to generate more edges.
Graphs <ER-x> signify input graphs generated randomly with $x$ vertices.
For example, ER-500000 signifies an input graph generated randomly with 500000 vertices.
Similarly graphs <Bipartite-x-y> signify Bipartite input graphs generated randomly with $\left<x,y\right>$ vertices.

Some of the input graphs were artificially generated by using the Erd\H{o}s R\'enyi model of random graphs~\cite{ER1959}, while real world graphs were obtained from the SNAP collection of large networks~\cite{SNAP}. 
Among the random graphs, we had both general random graphs as well as bipartite random graphs.
}


We seek to answer the following questions from the experiments.
\begin{enumerate}
\itemsep0em
\item
What is the relative performance of the different methods for MBE?
\item
How do these methods scale with increasing number of reducers?
\item
How does the runtime depend on the input size and the output size?
\end{enumerate}

\begin{table*}
\caption{Mean and Standard Deviation computation of all 100 reducer 
runtimes for Algorithms CD0, CD1 and CD2.
The analysis is done for the Reducer of the last MapReduce round
as it performs the actual Depth First Search computation. }
\label{table:stdeva}
\centering
\begin{tabular}{c c c c} 
loc-BrightKite-0.6 & CD0 & CD1 & CD2 \\ [0.5ex]
\hline\hline 
Average & 1005.94 & 631.92 & 625.14 \\
Variance & 3470135.82 & 256859.25 & 302764.97 \\
Standard Deviation & 1862.83 & 506.81 & 550.24 \\
\hline
\end{tabular}
\begin{tabular}{c c c c} 
ego-Facebook-0.6 & CD0 & CD1 & CD2 \\ [0.5ex]
\hline\hline 
Average & 479.37 & 380.32 & 422.7 \\
Variance & 473875.57.82 & 80146.79 & 273575.95 \\
Standard Deviation & 688.39 & 283.10 & 523.04 \\
\hline
\end{tabular}
\end{table*}

Figure~\ref{fig:runtime} plots the runtime data for the algorithms
mentioned in Table~\ref{table:inputs}. All data used for these plots
was generated with $100$ reducers.
\emph{The runtime data given for the Parallel Algorithms include the 
time required to run all MapReduce rounds including time required
to construct 2--neighborhood etc}.

\paragraph{Impact of the Pruning Optimization.} From
Figure~\ref{fig:runtime}, we can see that the optimizations to basic
DFS clustering through eliminating redundant work make a significant
impact to the runtime for all input graphs. For instance, in
Figure~\ref{fig:runtime-d}, on input graph email--EuAll--0.6 CD0,
which incorporates these optimizations, runs 9 times faster than CDFS,
the basic clustering approach without reducing redundant work.

\paragraph{Impact of Load Balancing} From Figure~\ref{fig:runtime}, we also
observe that for graphs on which the algorithms do not finish very
quickly (within 200 seconds), load balancing helps significantly. In
Figure~\ref{fig:runtime-d}, for graph email--EuAll--0.6, the Load
Balancing approaches (CD1 and CD2) are 7 -- 7.4 times faster than the
Algorithm without Load Balancing, but optimized to reduce redundant
work (CD0). Other examples include Figures~\ref{fig:runtime-e}, where
for input graph loc--BrightKite--0.6, CD1 was 4.5 times faster than
CD0 and CD2 was about 3.8 times faster. For some graphs, such as
email-EuAll-0.4 and web-NotreDame-0.8, CD0 failed to complete even
after 12 hours, but CD1 and CD2 were successful in processing them
within 2 hours. This shows that there must be significant imbalance in
the load, with subproblems on some dense clusters being much more
expensive to process than others. This shows that {\bf for most input
graphs, the versions optimized through both load balancing and
reducing redundant work worked the best overall}.

However, for graphs that are quickly processed, the load balancing
performs slightly slower than CD0 (see
Figure~\ref{fig:runtime-a}). This can be explained by the additional
overhead of load balancing (an extra round of MapReduce), which does
not payoff unless the work done at the DFS step is significant.

We tried two Load Balancing approaches, one based on the vertex degree
and the other on the size of the 2-neighborhood of the vertex. From
Figure~\ref{fig:runtime} we can observe no one approach was
consistently better than the other, and the performance of the two
were close to each other. For some input graphs, like Email-EuAll-0.4,
the 2-neighborhood approach (CD2) fared better than the degree
approach (CD1), whereas for some other input graphs like
web-NotreDame-0.8, the degree approach fared better.

Finally, to further analyze the impact of Load Balancing,
we also calculated the Mean and Standard Deviation of the run
time of each of the 100 reducers of the Reducer Algorithm~\ref{NCR_DFS} 
for CD0 and Reducer Algorithm~\ref{NCR_DFS_Three} 
for Algorithms CD1 and CD2 respectively. We present results of this
analysis for input graphs loc-BrightKite-0.6 and ego-Facebook-0.6.
Table~\ref{table:stdeva} shows the mean and standard deviation
analysis for the above-mentioned input graphs.
{\bf Observe that the Load Balanced Algorithms CD1 and CD2
have a much less standard deviation for reducer runtimes than CD0.}

\paragraph{Consensus versus Depth First Search} 
Both clustering consensus as well as parallel consensus were much
slower than the clustering DFS approaches. In all instances except for
very small input graphs, clustering consensus was 6-11 times slower
than CD1 and CD2 or worse (in many cases, clustering consensus did not
finish within 12 hours while CD1 and CD2 finished within 1-2 hours).
Further, direct parallel consensus, which uses a different
parallelization strategy as explained in Section~\ref{sec:parallel},
was 13 to 100 times slower than clustering consensus. These numbers
are not plotted in the figures. This shows that {\bf Clustering based
on DFS is the method of choice for parallel MBE.}

\remove{
other instances Except for some random graphs, Clustering Consensus approach took 6-11 times more time to finish
than the Depth First Search approach.  Also, for some input graphs,
the clustering consensus approach failed with Java heap error.  The
direct parallelization of the consensus approach takes an order of
magnitude more time than CD1 or CD2. Hence the results from the
consensus experiments are not shown in the paper.
}


\paragraph{Scaling with Number of Reducers}
We measured how Algorithms CD1 and CD2 scaled with the number of
reducers. In Figure~\ref{fig:scaling} we plot the runtime of CD1 and
CD2 with and increase in the number of reducers.  In
Figure~\ref{fig:speedup}, we also plot the speedup, defined as the
ratio of the time taken with $1$ reducer to the time taken with $r$
reducers, as a function of the number of reducers $r$. We observe that
the runtime decreases with increasing number of reducers, and further,
the algorithms achieve near-linear speedup with the number of
reducers.  This data shows that the algorithms are scalable and may be
used with larger clusters as well.

\paragraph{Relationship to Output Size:}
We also measured the runtime of the algorithms with respect to the
output size. We define the output size of the problem as the sum of
the number of edges for all enumerated maximal bicliques.
Figure~\ref{fig:runtime_vs_output} shows the runtime of algorithms
CD0, CD1, and CD2 as a function of the output size. This data is only
constructed for random graphs, where the different graphs considered
are generated using the same model, and hence have very similar
structure. We observe that {\bf the runtime increases almost linearly
with the output size for all three algorithms CD0, CD1, and CD2}.

With real world graphs, this comparison does not seem as appropriate,
since the different real worlds graphs have completely different
structures; however, we observed that Algorithms CD1 and CD2 have a
near-linear relationship with the output size, even on real world graphs.

\remove{
real--world
graphs from different backgrounds have different structures and hence
it would not be a fair comparison.  Same holds true for a bipartite
graph.  However, both the Algorithms CD1 and CD2 have a linear
increase in runtime with the output size even for the real--world
graphs.

We can make the following observations from our experiments:

\begin{itemize}
\item From Figure~\ref{fig:runtime_vs_output} we can observe that the runtime for the algorithms 
increase almost linearly with with the Output Size. This is true for all the three algorithms CD0, CD1, and CD2. This plot was made only for general random graphs as they follow the same uniform structure.
However, the data available in Table~\ref{table:inputs} clearly shows that the same holds true for real-world graphs when using algorithms CD1 and CD2.

\end{itemize}
}

\remove{
We observe that the
runtime of both the algorithms decrease with the increase in the
number of reducers for all the different input graphs measured.  This
gives us evidence that the algorithms scale well with the increase in
number of reducers.  This is important as this signifies that the
algorithms can be used in small as well as large clusters.
}

\paragraph{Large Maximal Bicliques:}
Finally we measured the runtime of Algorithm CD1 with respect 
to the size threshold $s$. Figure~\ref{fig:large_maxbicliques} 
shows the runtime with respect to various size thresholds 
ranging from 1 to 5. 
{\bf We can observe that the runtime decreases with increase in 
the size threshold $s$.}
For instance, in Figure~\ref{fig:large_maxbicliques-a}, we can observe
for the graph email-EuAll-0.4, the time taken to enumerate 
maximal bicliques larger than size 5, is about 5 times less and the 
time taken for size threshold 3 is about 2.1 times less than that to 
enumerate all maximal bicliques.
Similarly, we can observe from Figure~\ref{fig:large_maxbicliques-b},
that enumeration times are 3.8 and 1.6 times less for size thresholds
5 and 3 respectively for graph web-NotreDame-0.8.

\section{Conclusion}
\label{sec:conclusion}

Mining maximal bicliques is a fundamental tool that is useful in
uncovering dense relationships within data that can be represented as
a graph. We presented scalable parallel methods for mining maximal
bicliques from a large graph. 
We presented a basic clustering framework for parallel enumeration of
bicliques. On top of this, we presented two optimizations, one for
reducing redundant work, and the other for improving load balance,
both of which significantly improved the observed performance. These
algorithms scale well with increasing numbers of reducers.
To our knowledge, our work is the first
to successfully enumerate bicliques from graphs of this size; previous
reported results were mostly sequential methods that worked on much
smaller graphs.

The following directions are interesting for exploration (1)~How does
this approach perform on even larger clusters, and consequently,
larger input graphs? What are the bottlenecks here? and (2)~Are there
further improvements in search space pruning, and load balancing?

\remove{
We present multiple parallel algorithms for the enumeration of all
maximal bicliques in a large graph. From the experimental results we
observe that both show output sensitive behavior. Further the
clustering approach shows much better performance than direct
parallelization. Although we could run our implementations
successfully on large and sparse graphs, it seemed difficult for
graphs with higher density.
}

\bibliographystyle{plainnat}
\bibliography{references}

\end{document}